\documentclass[12pt, draftclsnofoot, onecolumn]{IEEEtran}
\usepackage{graphicx}
\usepackage{amsthm}
\usepackage{epsfig}
\usepackage{latexsym}
\usepackage{amsfonts}
\usepackage{here}
\usepackage{rawfonts}
\usepackage[latin1]{inputenc}
\usepackage[T1]{fontenc}
\usepackage{calc}
\usepackage{capitalgreekitalic}
\usepackage{url}
\usepackage{enumerate}
\usepackage{color}
\usepackage[tbtags]{amsmath}
\usepackage{amssymb}
\usepackage{upref}
\usepackage{epic,eepic}
\usepackage{times}
\usepackage{dsfont}
\usepackage{comment}
\usepackage{cite}
\usepackage{bbm} 
\newtheorem{corollary}{Corollary}
\newtheorem{theorem}{\bf Theorem}

\usepackage{dsfont}
\newcounter{step}
\newlength{\totlinewidth}
  {\end{list}%
  \rule{\linewidth}{1pt}}
\newcounter{substep}

  {\end{list}}

\newlength{\aligntop}
\newlength{\alignbot}
 \makeatletter
\renewenvironment{align}{%
  \vspace{\aligntop}
  \start@align\@ne\st@rredfalse\m@ne
}{%
  \math@cr \black@\totwidth@
  \egroup
  \ifingather@
    \restorealignstate@
    \egroup
    \nonumber
    \ifnum0=`{\fi\iffalse}\fi
  \else
    $$%
  \fi
  \ignorespacesafterend%
  \vspace{\alignbot}\par\noindent
} \makeatother

\IEEEoverridecommandlockouts

\usepackage{algorithm}%,algpseudocode}
\usepackage{algorithmic}
\usepackage{subfigure}
\begin{document}
\clearpage
\title{\huge Data Correlation-Aware Resource Management in Wireless Virtual Reality (VR): An Echo State Transfer Learning Approach \vspace*{-0em}}
\author{{Mingzhe Chen\IEEEauthorrefmark{1}}, Walid Saad\IEEEauthorrefmark{2}, Changchuan Yin\IEEEauthorrefmark{1}, and M\'erouane Debbah\IEEEauthorrefmark{3}\vspace*{0.5em}\\ 
%M\'erouane Debbah\IEEEauthorrefmark{3}, and Choong-Seon Hong\IEEEauthorrefmark{4}\vspace*{0em}\\
\authorblockA{\small \IEEEauthorrefmark{1}Beijing Key Laboratory of Network System Architecture and Convergence,\\ Beijing University of Posts and Telecommunications, Beijing, China 100876,\\ Emails: \protect\url{chenmingzhe@bupt.edu.cn}, \protect\url{ccyin@ieee.org.} \\
\IEEEauthorrefmark{2}Wireless@VT, Bradley Department of Electrical and Computer Engineering, Virginia Tech, Blacksburg, VA, USA, Email: \protect\url{walids@vt.edu.}\\
\IEEEauthorrefmark{3}\small Mathematical and Algorithmic Sciences Lab, Huawei France R \& D, Paris, France, \\Email: merouane.debbah@huawei.com.\\
%\IEEEauthorrefmark{4}\small Department of Computer Science and Engineering, Kyung Hee University, Yongin, South Korea, \\Email: \protect\url{cshong@khu.ac.kr.}\\
}\vspace*{-2em}
%\thanks{This work was supported in part by the National Natural Science Foundation of China under Grants 61671086 and 6162910, by the U.S. National Science Foundation under Grants IIS-1633363, CNS-1460316, and CNS-1617896 and by the ERC Starting 
%Grant 305123 MORE (Advanced Mathematical Tools for Complex Network Engineering).}
 }

\maketitle
{\renewcommand{\thefootnote}{\fnsymbol{footnote}}
\footnotetext{A preliminary version of this work was published in \cite{ESTL2017Chen}.}}

\vspace{0cm}
\begin{abstract}
Providing seamless connectivity for wireless virtual reality (VR) users has emerged as a key challenge for future cloud-enabled cellular networks. In this paper, the problem of wireless VR resource management is investigated for a wireless VR network in which VR contents are sent by a cloud to cellular small base stations (SBSs). The SBSs will collect tracking data from the VR users, over the uplink, in order to generate the VR content and transmit it to the end-users using downlink cellular links. For this model, the data requested or transmitted by the users can exhibit correlation, since the VR users may engage in the same immersive virtual environment with different locations and orientations. As such, the proposed resource management framework can factor in such spatial data correlation, so as to better manage uplink and downlink traffic.
This potential spatial data correlation can be factored into the resource allocation problem to reduce the traffic load in both uplink and downlink.
In the downlink, the cloud can transmit $360^\circ$ contents or specific visible contents (e.g., user field of view) that are extracted from the original $360^\circ$ contents to the users according to the users' data correlation so as to reduce the backhaul traffic load. In the uplink, each SBS can associate with the users that have similar tracking information so as to reduce the tracking data size. 
 This data correlation-aware resource management problem is formulated as an optimization problem whose goal is to maximize the users' successful transmission probability, defined as the probability that the content transmission delay of each user satisfies an instantaneous VR delay target. To solve this problem, a machine learning algorithm that uses echo state networks (ESNs) with transfer learning is introduced. By smartly transferring information on the SBS's utility, the proposed transfer-based ESN algorithm can quickly cope with changes in the wireless networking environment due to users' content requests and content request distributions. Simulation results demonstrate that the developed algorithm achieves up to 15.8\% and 29.4\% gains in terms of successful transmission probability compared to Q-learning with data correlation and Q-learning without data correlation. %The results also show that the proposed algorithm has a faster convergence time than Q-learning and can guarantee low delays.
 
 \end{abstract}

\vspace{0cm}
{\small \emph{Index Terms}--- virtual reality; resource allocation; echo state networks; transfer learning.}
%\renewcommand{\thefootnote}{\fnsymbol{footnote}}
%\footnotetext{A preliminary version of this work \cite{QoSmodelR} was submitted to IEEE GLOBECOM}

\section{Introduction}
Augmented and virtual reality applications will be an integral component of tomorrow's wireless networks \cite{ bacstuug2016towards }. Indeed, it is envisioned that by using wireless virtual reality (VR) services, users can engage in unimaginable virtual adventures and games within the confines of their own home. {\color{black}As a VR device is operated over a wireless network, the VR users must send tracking information that includes the users' locations and orientations to the small base stations (SBSs) and, then, the SBSs will use the tracking information to construct $360^\circ$ images and send these images to the users. Therefore, for wireless VR applications, the uplink and downlink transmissions must be jointly considered. Moreover, in contrast to traditional video that consists of $120^\circ$ images, a VR video consists of high-resolution $360^\circ$ vision with three-dimensional surround stereo. This new type of VR video requires a much higher data rate than that of traditional mobile video.} In consequence, to enable a seamless and pervasive wireless VR experience, it is imperative to address many wireless networking challenges that range from low latency and reliable networking to effective communication, computation, and resource management \cite{rosedale2017virtual}.

To address these challenges, a number of recent works on wireless VR recently appeared such as in \cite{bacstuug2016towards,rosedale2017virtual,ahn2017delay,singh2017high,chakareski2017vr,VROWNchen,kasgari2018human,ESTL2017Chen,park2018urllc,sun2018communication,8319985,8377419}. In \cite{bacstuug2016towards} and \cite{rosedale2017virtual}, qualitative surveys are provided to motivate the use of VR over wireless networks and to present the associated opportunities. The authors in \cite{ahn2017delay} proposed an efficient wireless VR communication scheme using a wireless local area network.
The work in \cite{singh2017high} studied the problem of VR tracking and positioning. However, the works in \cite{ahn2017delay} and \cite{singh2017high} only analyze a single VR metric such as tracking accuracy and do not develop a specific wireless-centric VR model. In \cite{chakareski2017vr}, the authors proposed a new algorithm for cached content replacement to minimize transmission delay.
 The work in \cite{VROWNchen} introduced a model for wireless VR services that takes into account the tracking accuracy, processing latency, and wireless transmission latency and, then, developed a game-theoretic approach for VR resource management. The authors in \cite{kasgari2018human} studied the resource allocation problem with a brain-aware QoS constraint. In \cite{park2018urllc}, the authors investigated the problem of concurrent support of visual and haptic perceptions over wireless cellular networks. In \cite{sun2018communication}, the authors developed a framework for mobile VR delivery by leveraging the caching and computing capabilities of mobile VR devices in order to alleviate the traffic burden over wireless networks. A communications-constrained mobile edge computing framework is proposed in \cite{8319985} to reduce wireless resource consumption. The authors in \cite{8377419} presented a new scheme for proactive computing and millimeter wave transmission for wireless VR networks.
 However, the works in \cite{chakareski2017vr,VROWNchen,kasgari2018human,park2018urllc,sun2018communication,8319985,8377419} ignore the correlation between the data of VR users. {\color{black}In fact, VR data (tracking data or VR image data) pertaining to different users can be potentially correlated because the users share a common virtual environment. For example, when the VR users are watching an event from different perspectives, the cloud has to only send one ${360^ \circ }$ image to the SBSs who, in turn, can rotate the image and send it to the various users. For such scenarios, one can reduce the traffic load on the cellular network, by exploiting such correlation of views among users.}
 {{Note that, in \cite{ESTL2017Chen}, we have studied the problem of data correlation-aware resource allocation in VR networks}. However, our work in \cite{ESTL2017Chen} considered only data correlation among two users and it relied on a very preliminary model that does not consider $360^\circ$ and visible contents transmission for the cloud. In practice, for a given user, a $360^\circ$ VR content can be divided into visible and invisible components. A \emph{visible content} is defined as the component of a $360^\circ$ content that is visible (in the field of view) to a given user. Moreover, our work in \cite{ESTL2017Chen} used a more rudimentary learning algorithm for resource allocation.}

{\color{black}The main contribution of this paper is to address the resource allocation problem for wireless VR networks while taking into account potential correlation among users. We introduce a novel model and associated solution approach using ESN-based transfer learning that enable SBSs to effectively allocate the uplink and downlink resource blocks to the VR users considering the data correlation among the uplink tracking information data and downlink VR content data so as to maximize VR users' successful transmission probabilities. 
 To our best knowledge, \emph{this is the first work to use ESNs as transfer learning for data correlation-aware resource block management and visible content transmission in wireless VR networks.} The primary contributions of this work can thus be summarized as follows:
 
 \begin{itemize}
\item We introduce a new model for VR in which the cloud can transmit $360^\circ$ or visible contents to the SBSs and the SBSs will transmit the visible contents to VR users. Meanwhile, the users will cooperatively transmit tracking information to the SBSs so as to enable them to extract the visible contents from original $360^\circ$ contents. %The distinction between $360^\circ$ and visible contents is novel and has not been exploited in any prior works on wireless.  
% jointly capture the downlink and uplink transmission delay, backhaul transmission delay, and computational time thus effectively quantifying the VR delay for all users in a wireless VR network.
\item We then investigate how the cloud can \emph{jointly} optimize uplink and downlink data transmission and resource block allocation. We formulate this joint content transmission and resource block allocation problem as an optimization problem. The goal of this optimization problem is to maximize the users' successful transmission probability. 
 
\item To address this problem, we propose a transfer learning algorithm \cite{5288526} based on ESNs \cite{chen2017machine}. This algorithm is able to smartly transfer information on the learned successful transmission probability across time so as to adapt to the dynamics of the wireless environment due to factors such as changes in the users' data correlation and content request distribution.

\item We analyze how the cloud determines the transmission format ($360^\circ$ or $120^\circ$ content) of each VR content that is transmitted over the cloud-SBSs backhaul links. Analytical results show that the transmission format of each VR content depends on the number of users that request different visible contents, the data correlation among the users, and the backhaul data rate of each VR user.

\item We perform fundamental analysis on the gain, in terms of the successful transmission probability, resulting from the changes of resource block allocation and  content transmission format. This analytical result can provide guidance for action selection in the proposed machine learning approach.  

\item Simulation results demonstrate that our proposed algorithm can achieve, respectively, 15.8\% and 29.4\% gains in terms of the total successful transmission probability compared to Q-learning with data correlation and Q-learning without data correlation. {The results also show that the proposed transfer learning algorithm needs 10\% and 14.3\% less iterations for convergence compared to Q-learning with data correlation and Q-learning without data correlation. }  
\end{itemize}}

 The rest of this paper is organized as follows. The system model and problem formulation are presented in Section \uppercase\expandafter{\romannumeral2}. The ESN-based transfer learning algorithm for resource block allocation is proposed in Section \uppercase\expandafter{\romannumeral3}. In Section \uppercase\expandafter{\romannumeral4}, numerical simulation results are presented and analyzed. Finally, conclusions are drawn in Section \uppercase\expandafter{\romannumeral5}.               

% most immersive experience witness thus far.  
\section{System Model and Problem Formulation}\label{se:system}
%\begin{figure}[!t]
%  \begin{center}
%   \vspace{0cm}
%    \includegraphics[width=7cm]{figure1b1new.eps}
%    \vspace{-0.3cm}
%    \caption{\label{figure1b} The components of one VR view. Here, the $360^\circ$ image for each eye is slightly different. These different images generate a 3D image for each VR user.}
%  \end{center}\vspace{-0.3cm}
%\end{figure}
%
%\begin{figure}[!t]
%  \begin{center}
%   \vspace{0cm}
%    \includegraphics[width=7cm]{imag360.eps}
%    \vspace{-0.3cm}
%    \caption{\label{figure1a}The images taken from VR device and camera (the top $360^\circ$ image is a VR image and the bottom $120^\circ$ image is a mobile phone image). $360^\circ$ image enables the VR users to have a vision without any dead spots and, hence, can be immersed in a surrounded virtual reality environment. Here, the dead spot corresponds to an area that the user cannot observe.}
%  \end{center}\vspace{-0.3cm}
%\end{figure}

 \begin{figure}[!t]
  \begin{center}
   \vspace{0cm}
    \includegraphics[width=10cm]{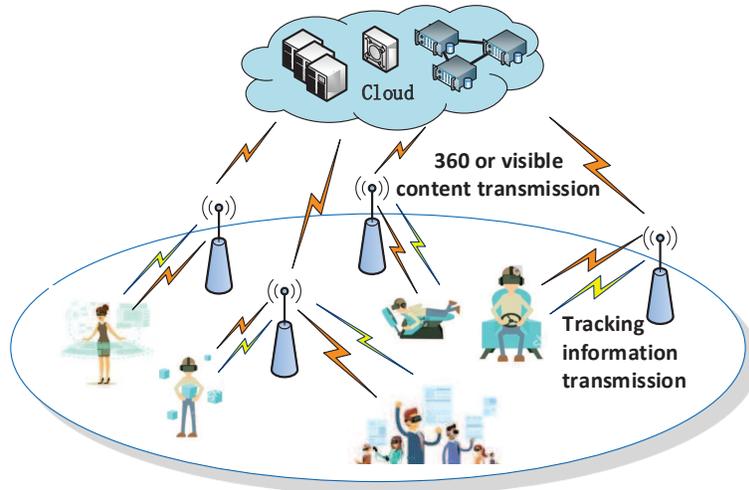}
    \vspace{-0.3cm}
    \caption{\label{systemmodel} The architecture of a VR network that consists of the cloud, SBSs, and VR users.}
  \end{center}\vspace{-0.5cm}
\end{figure}

Consider the wireless network shown in Fig. \ref{systemmodel} that is composed of a set $\mathcal{K}$ of $K$ SBSs that serve, in both uplink and downlink, a set $\mathcal{U}$ of $U$ users. In this network, uplink transmissions are used to carry tracking information (e.g., VR user location and orientation) from the users to the network. Meanwhile, downlink transmissions will carry the actual VR content (e.g., images) to the users. A \emph{capacity-constrained} backhaul connection is considered between the SBSs and the cloud. To serve the VR users, the SBSs will use the cellular band.
The cloud can directly transmit the $360^\circ$ contents to the SBSs or extract the visible contents from the $360^\circ$ contents and transmit them to the SBSs, as shown in Fig. \ref{model}. In our model, a 360$^\circ$ VR content that consists of $360^\circ$ images while a visible content is composed of $120^\circ$ horizontal and $120^\circ$ vertical images\footnote{\color{black}Here, the degree of the visible contents of a user depends on the devices that the users used for engaging in the VR applications. 120$^\circ$ is typically used with HTC Vive \cite{htc}.}, as shown in Fig. \ref{visiblecontent}. If the cloud wants to transmit visible contents to the SBSs, it needs to acquire the users' tracking information from the SBSs.
Let $G_{120^\circ}$ be the data size of each visible content and $G_{360^\circ}$ be the data size of each $360^\circ$ content. %We also assume that the VR users that request the same content will request the same $360^\circ$ contents but different visible contents. 
To reduce the traffic load over SBS-users links, the SBSs will only transmit visible contents to the users, as shown in Fig. \ref{model}.

 \begin{figure}[!t]
  \begin{center}
   \vspace{0cm}
    \includegraphics[width=10cm]{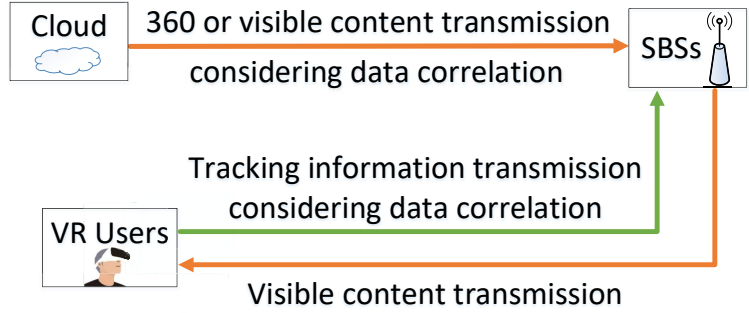}
    \vspace{-0.3cm}
    \caption{\label{model} The content and tracking information transmissions in a VR network.}
  \end{center}\vspace{-0.4cm}
\end{figure}

  \begin{figure}
  \begin{center}
   \vspace{0cm}
    \includegraphics[width=6cm]{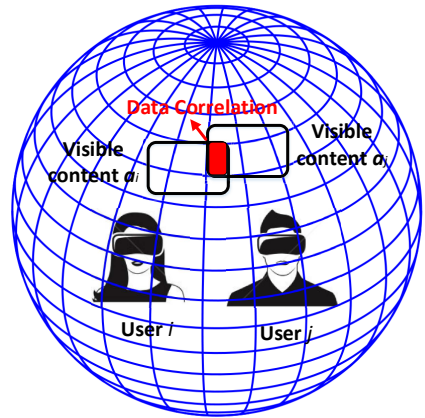}
    \vspace{-0.3cm}
    {\color{black}\caption{\label{visiblecontent} An illustrative example of $360^\circ$ and visible contents. Here, $a_i$ and $a_j$ represent the visible content requested by users $i$ and $j$, respectively. The red region shows the similarities between $a_i$ and $a_j$.}}
  \end{center}\vspace{-0.7cm}
\end{figure}

We adopt an orthogonal division multiple access (OFDMA) scheme. {\color{black}The SBSs use a set $\mathcal{V}$ of $V$ resource blocks for the uplink and a set $\mathcal{S}$ of $S$ resource blocks for the downlink.} We assume that each SBS has a circular coverage area of radius $r$. {\color{black}We also assume that the resource blocks of each SBS are all allocated to the users.} We also consider that each SBS will allocate all of its resource blocks to its associated users. {\color{black}Let $\boldsymbol{q}_{in}=\left[q_{in,1},\ldots, q_{in,N_C}  \right]$ be the content request distribution \cite{8382257} of user $i$ during a period $n$ where $q_{in,k}$ represents the probability that user $i$ requests content $k$ and $N_C$ is the total number of contents that each user can request. For different periods, the content request distribution of each user will be different. We also assume that the users will immediately request a new VR content once they completely receive one content.}

\begin{table}\label{ta:notation}\footnotesize
  \newcommand{\tabincell}[2]{\begin{tabular}{@{}#1@{}}#2\end{tabular}}
\renewcommand\arraystretch{1}
 \caption{
    \vspace*{-0.3em}List of notations}\vspace*{-1.5em}
\centering  
\begin{tabular}{|c||c|c||c|}% ±íÊŸž÷ÁÐÔªËØ¶ÔÆë·œÊœ£¬left-l,right-r,center-c
%\hline
%\textbf{Type} & \textbf{Power} & \textbf{Cache} & \textbf{Moving}&\textbf{Path loss} & \textbf{Frequency bands}&\textbf{Interference} \\
\hline
Notation & Description & Notation & Description\\
\hline
$U$ & Number of VR users & $h_{ij}^k$ & Path loss between user $i$ and SBS $j$\\
\hline
$K$ & Number of SBSs & $B$ & Bandwidth of each resource block  \\
\hline
$P_U$& Transmit power of each user  & $N_{aik}$ & Data size of different pixels between visible contents $a_i$ and $a_k$\\% $R_{Ck}$ & Rate from the cloud to UAV $k$ \\
\hline
  $\boldsymbol{y}_{j}'\left(t\right)$& Output b) of ESN & $C_{aik}$ & Data correlation between user $i$ and $k$ \\
\hline
 $\mathcal{C}_a$ & Set of data correlation&$\upsilon^2$ & Variance of the Gaussian noise  \\
\hline
$\mathcal{U}$ & Set of VR users & $g_{ja}$ & Content transmission format\\
\hline
$\mathcal{K}$ & Set of SBSs & $\sigma_{ij}$ & Covariance of tracking information \\
\hline
 $\boldsymbol{y}_{j}\left(t\right)$& Output a) of ESN   &  $V_i^B $ & Fronthaul transmission rate of each user $i$ \\
\hline
  $c_{ij}$ & Downlink data rate of user $i$&$S$ & Number of downlink resource blocks\\
\hline
 $P_{B}$& Transmit power of each SBS & ${D _{ij}}$ & Delay of content transmission\\
\hline
$N_{W}$ &Number of neurons in ESN &${{D}_{ij}^\textrm{U}} $ & Delay of tracking information transmission\\
\hline
 $G_{120^\circ}$ & Data size of each visible content & $Q_{it}$ & Successful transmission of user $i$ at time $t$ \\
\hline
$G_{360^\circ}$& Data size of each $360^\circ$ content & ${\mathbbm{P}_i}$ & Successful transmission probability of user $i$\\
\hline
$c_{ij}^\textrm{UL}$& Uplink data rate of user $i$ &$M\left( {g}_{ja} \right)$& Data size of content $a$ with transmission format ${g}_{ja} $\\
\hline
$\boldsymbol{x}_{j}\left(t\right)$& Input of ESN at time $t$&$K_i\left(\sigma_i^{\max}\right)$& Data size of tracking information\\
\hline
$\boldsymbol{a}_j$&An action of SBS $j$& $ V$ & Number of uplink resource blocks \\
\hline
$N_{ja}$& Number of actions of SBS $j$&$d_{ij}$ & Distance between BS $j$ and user $i$\\
\hline
% $ V$ & Number of beams of each BS &$Y$ & Number of locations and orientations that each ESN can predict\\
%\hline
%$x_{t,i},y_{t,i}$ & Coordinates of users   & $P_B$ & Transmit power of the BBUs\\
%\hline
\end{tabular}
\end{table}

 %We assume that each SBS's coverage is a circular area with radius $r$. We also assume that each SBS will allocate all of the resource blocks to the users located in its coverage range. %Table \uppercase\expandafter{\romannumeral1} provides a summary of the notations used hereinafter.

%Fig. \ref{figure2} shows the architecture of immersive VR application network. In this architecture, the transmission between the SBSs and the VR users occurs over wireless links. The tracking information that is collected by the VR sensors placed at VR user's helmet or near VR user is also transmitted via the wireless links. Here, tracking pertains to the fact that the immersive VR applications need to know very accurate localization of each user including the positions, orientation, and eye movement (i.e., gaze tracking). Hence,   %For this system, we consider , which splits the licensed band into equal spectrum bands for the downlink and uplink. 
 
 % Note that we focus on single carrier system. The bandwidth of each SBS is limited and $F_j$ represents the bandwidth of each SBS $j$.
   
\subsection{Proposed Model for VR Data Correlation}
\subsubsection{Model for Data Correlation in the Downlink}    
{\color{black}In the downlink, the cloud will transmit VR contents to the SBSs. For VR applications, each VR content consists of $360^\circ$ images, which means that the cloud must transmit all of the surrounding virtual environment information to each SBS. When the users engage in the same VR activity or play the same immersive games, they will share the same virtual environment information thus making their downlink data correlated.}
 %When using wireless VR services, users are often engaged in a common immersive game, but they have different views, orientations, and locations. 
%For such scenarios, 
The network can better manage its backhaul traffic if it exploits the correlation of data for users that are engaged in the same VR game or activity.
 For instance, as the users are engaged in a common virtual football or basketball game, the cloud can directly transmit the entire $360^\circ$ contents to the SBSs. Then, the SBSs can extract the unique visible contents from the  $360^\circ$ data and transmit them to the users.    
% extract the difference between the VR contents of these users and will need to only transmit to an SBS the data that is unique to each user. 
 However, if the users are engaged in different VR activities, then, their VR data correlation between the users is low and, therefore, different VR contents must be sent to those users.
 
To properly define VR data correlation, {\color{black}we consider that the VR users that request the similar $360^\circ$ content will request different visible contents. This is because the users may have different location and orientation.} In consequence, they may observe different components of a $360^\circ$ content. Let $a_i$ be the visible content that is extracted from $360^\circ$ content $a$ and requested by user $i$. {\color{black}Let $C_{aik}$ be the fraction of the same pixels between visible contents $a_i$ and $a_k$ that users $i$ and $k$ request, respectively.} %Hereinafter, the users that request content $a$ is referred to the users that request different visible contents $a$ for simplicity. 
% The data size of different pixels between visible contents $a_i$ and $a_k$ that users $i$ and $k$ request is $N_{aik}$. If users $i$ and $k$ request the same visible content $a_i=a_k$, $N_{aik}=0$. Here, $N_{aik}$ is computed by the cloud using motion search \cite{yang2002computation} or other image processing schemes. Then, for visible contents $a_i$ and $a_k$, the data correlation between user $i$ and user $k$ can be defined as follows:
%\begin{equation}\label{eq:datacu}
%{C _{aik}} = 1-\frac{{{N_{aik}}}}{{{G_{120^\circ}}}}.
%\end{equation}
%Indeed, (\ref{eq:datacu}) captures the similarities between the visible contents $a_i$ and $a_k$ that users $i$ and $k$ request. 
In this context, as users $i$ and $k$ are connected to the same SBS, the cloud needs to only transmit $G_{120^\circ}\left(2-{C _{aik}}\right)$ Mbits data of visible contents to that SBS. Let ${\mathcal{C}_a^n}$ be the set of data correlation among any $n$ users that {request different visible contents extracted from $360^\circ$ content $a$}. For example, for a given SBS that is serving three users (users 1, 2, and 3), ${\mathcal{C}_a^2}=\left\{ {C _{a12}}, {C _{a13}}, {C _{a23}} \right\}$ and ${\mathcal{C}_a^3}=\left\{ {C _{a123}} \right\}$. The cloud can select the appropriate content format (visible or $360^\circ$ content) for each content transmission. Let $g_{ja} \in \left\{120^\circ, 360^\circ \right\}$ be the content transmission format. $g_{ja}=120^\circ$ implies that the cloud will transmit visible contents that are extracted from $360^\circ$ content $a$ to SBS $j$, otherwise, the cloud will transmit $360^\circ$ content $a$ to SBS $j$.   
%
%For the cloud, the choice of visible or $360^\circ$ content transmission can be given by the following theorem:
%
%\begin{theorem}\label{th:1}
%\emph{Given the set $\mathcal{U}_{ja}$ of ${U}_{ja}$ users that request different visible contents extracted from $360^\circ$ content $a$, and the set of data correlation $\mathcal{C}_{a}=\bigcup\limits_{n = 2}^{{U_{ja}}} {\mathcal{C}_a^n} $, the transmission format of content $a$ can be given by:
%\begin{itemize}
%\item If $G_{360^\circ} \geqslant  L_a\left(\mathcal{C}_{a}\right)$, $g_{ja}=120^\circ$.
%\item If $G_{360^\circ} < L_a\left(\mathcal{C}_{a}\right)$, $g_{ja}=360^\circ$.
%\end{itemize}
%Here, $L_a\left(\mathcal{C}_{a}\right)= G_{120^\circ} \left( {{U_{ja}} - \sum\limits_{n = 2}^{{U_{ja}}} {\sum\limits_{{C_a} \in \mathcal{C}_a^n} {{{\left( { - 1} \right)}^{n - 1}}{C_a}} } } \right)$.
%}
%
%\end{theorem}
%\begin{proof} See Appendix A.
%\end{proof}
% Theorem \ref{th:1} shows that the choice of $360^\circ$ and visible content transmission depends on the data size of $360^\circ$ and visible contents, data correlation among the users, and the number of users that request the same content. When the data correlation among the users increases, the cloud prefers to transmit visible contents to the SBSs. In contrast, as the number of users that request the same content decreases, the cloud prefers to transmit the $360^\circ$ content.  
    
\subsubsection{Data Correlation Model for the Uplink}  

In the uplink, the tracking information is collected by the sensors that are located at the VR users' headsets. The VR user sensors need to scan their environment and send the information related to their environment to the SBSs. {\color{black}Therefore, the VR users' sensors will collect data from a similar environment and, hence, this data will be correlated~\cite{Wu:2017:DEU:3083187.3083210}. In consequence, the tracking information of wireless VR users will have data correlation.}
 %For uplink transmissions, the users must send tracking data, collected by their VR headset sensors or nearby VR sensors, to their serving SBSs. 
 {\color{black}To model each user's tracking information that is collected by the VR sensors, we adopt a Gaussian field model similar to the one use in \cite{cressie2015statistics}.} % As a result, we assume that the tracking information follows a Gaussian distribution.}
  Let the tracking data, $X_i$, gathered by each VR user $i$ be a Gaussian random variable with variance $\sigma_i^2$ and mean $\mu_i$. {\color{black}This model for each user's tracking information is constructed based on the historical tracking information collected by the SBSs. The SBSs can use the expectation maximization algorithms in \cite{Bishop2006Pattern} to determine the parameters of  the Gaussian field model.} In VR applications, observations from neighboring VR devices will often be correlated. {\color{black}For example, VR users that are located close in proximity or within a common location (VR theater or stadium) may request similar visible VR contents. Based on the model of tracking information, we employ the power exponential model \cite{vuran2004spatio} to model the spatial correlation of the VR tracking data since the power exponential model can capture how the distance between two users impacts data correlation.} Consequently, for any two VR users $i$ and $j$ located at a distance $d_{ij}$, the covariance $\sigma_{ij}$, will be given by \cite{8254516}:
 \begin{equation}\label{eq:sigma}
 {\sigma _{ij}} = {\mathop{\rm cov}} \left( {{X_i},{X_j}} \right) = {\sigma _i}{\sigma _j}{e^{ - {{d_{ij}^\alpha } \mathord{\left/
 {\vphantom {{d_{ij}^\alpha } \kappa }} \right.
 \kern-\nulldelimiterspace} \kappa }}},
\end{equation}    
 where $\alpha$ and $\kappa$ are parameters that capture how sensitive data correlation will be to distance variations. %{\color{black}Note that, for users that newly join the network and for which no historical data exists, the SBSs will not consider any correlation during the initial connection of the user, until it obtains enough data to explore correlation between this user and others.}
 
 \subsection{Transmission Model}
In the studied model, the cloud will first transmit the VR contents requested by the users to the SBSs. {Then, the SBSs will transmit the contents received from the cloud to its associated users}. {Meanwhile, the users will transmit their tracking information to the SBSs. If the cloud determines that it needs to transmit the visible contents to the SBSs, the SBSs must transmit the tracking information to to the cloud.} 
%In consequence, the tracking information must be processed at each corresponding SBS or at the cloud. In the downlink, we consider VR content transmission over backhaul and SBS-users links. 
However, in the uplink, we only consider the tracking information transmission over the wireless SBS-users links and ignore the delay of tracking information transmission over wired backhaul links. This is due to the large capacity of the wired backhaul and the relatively small data size of tracking information compared to the VR content data size.    
For each user $i$, the transmission rate of each VR content from the cloud to the SBS can be given by \cite{7438747}: % \cite{chen2016caching}:
\begin{equation}
 {V_i^{\textrm{B}}} = \frac{{{V^\textrm{B}}}}{U},
 \end{equation}
 where $V^\textrm{B}$ is the maximum downlink backhaul link rate, for all VR users.
 We assume that the backhaul rates for all users are equal and we do not consider backhaul transmission optimization. 
In a VR model, we need to capture the VR transmission requirements such as high data rate, low delay, and accurate tracking. Hence, we consider the transmission delay as the main VR QoS metric of interest. 
For user $i$ associated with SBS $j$, the downlink rate can be given by:
\begin{equation}\label{eq:cd}
{c_{ij}\left(\boldsymbol{s}_{ij}\right)} = \sum\limits_{k = 1}^{{S}} {s_{ij,k}B{{\log }_2}\left( {1 + \gamma _{ij,k}} \right)},
\end{equation}
where $\boldsymbol{s}_{ij} = \left[ {s_{ij,1}, \ldots ,s_{ij,{S}}} \right]$ represents the resource blocks vector with $s_{ij,k} \in \left\{ {0,1} \right\}$. $s_{ij,k}=1$ indicates that SBS $j$ allocates resource block $k$ to user $i$, and $s_{ij,k}=0$, otherwise. ${\gamma _{ij,k}=\frac{{{P_B}{h_{ij}^k}}}{{{\upsilon  ^2} + \sum\limits_{l \in \mathcal{K},l \ne j} {{P_B}{h_{il}^k}} }}}$ represents the signal-to-interference-plus-noise ratio (SINR) between SBS $j$ and user $i$ on resource block $k$. $B$ represents the bandwidth of each resource block and $P_B$ represents each SBS $j$'s transmit power. {\color{black} Moreover, $ \sum\limits_{l \in \mathcal{K},l \ne j} {{P_B}{h_{il}^k}} $ is the interference caused by other SBSs using resource block $k$ for VR content transmission}.
 $\upsilon^2$ represents the variance of the Gaussian noise and $h_{ij}^k=g_{ij}^kd_{ij}^{-\beta}$ represents the path loss between SBS $j$ and user $i$ with $g_{ij}^k$ being the Rayleigh fading parameter. $d_{ij}$ represents the distance between SBS $j$ and user $i$, and $\beta$ represents the path loss exponent. 
For user $i$ associated with SBS $j$, the delay of transmitting content $a$ from the cloud to user $i$ is:
\begin{equation}\label{eq:Dsg}
 D_{ij}\left(\boldsymbol{s}_{ij},{g}_{ja}\right) = \frac{{{ G_{120\circ} }}}{{{c_{ij}\left(\boldsymbol{s}_{ij}\right)}}}+\frac{{{M\left( {g}_{ja} \right)}}}{U_{ja}V_{i}^\textrm{B}},
 \end{equation}
  where $U_{ja}$ is the number of users that are associated with SBS $j$ and request visible contents $a$. $M\left( {g}_{ja}\right) $ is the data size of content $a$ must be transmitted from the cloud to SBS $j$. %is the data that user $i$ needs to construct a VR image during a period and $\phi _{i}^{\max}\!\!=\!\!\!\mathop {\max }\limits_{k\in \mathcal{U}_j, k \ne i}\! \left( {{\phi _{ik}}} \right)$ is the maximum downlink data correlation between user $i$ and other users associated with SBS $j$. Finding the maximum data correlation allows minimizing the downlink transmission data transmitted in the downlink and that will be used construct a VR image. 
The first term in (\ref{eq:Dsg}) is the time that SBS $j$ uses to transmit content $a$ to user $i$ and the second term in (\ref{eq:Dsg}) is the time that the cloud uses to transmit content $a$ to SBS $j$. {\color{black}From (\ref{eq:Dsg}), we can see that if several users request the same content, the cloud
will perform only one transmission over the backhaul using the sum backhaul data rate of these users.}
We let $P_U$ be the transmit power of each user and $B$ be the bandwidth of each uplink resource block. In this case, for user $i$ connected to SBS $j$, the uplink data rate can be given by:
\begin{equation}\label{eq:cv}
{c_{ij}\left(\boldsymbol{v}_{ij}\right)} = \sum\limits_{k = 1}^{{V}} {v_{ij,k}B{{\log }_2}\left( {1 + \gamma _{ij,k}^\textrm{u}} \right)},
\end{equation}
where $\boldsymbol{v}_{ij} = \left[ {v_{ij,1}, \ldots, v_{ij,{V}}} \right]$ represents the vector of uplink resource blocks that SBS $j$ allocates to user $i$ with $v_{ij,k} \in \left\{ {0,1} \right\}$. $v_{ij,k}=1$ indicates that SBS $j$ allocates uplink resource block $k$ to user $i$, {\color{black}$v_{ij,k}=0$}, otherwise.
${\gamma _{ij,k}^\textrm{u}=\frac{{{P_U}{h_{ij}^k}}}{{{\upsilon ^2} + \sum\limits_{l \in \mathcal{U}^k,l \ne j} {{P_U}{h_{il}^k}} }}}$ represents the SINR between  SBS $j$ and user $i$ on resource block $k$. Here, $\mathcal{U}^k$ is the set of VR users that transmit their tracking information over uplink resource block $k$. 
The time that user $i$ needs to transmit tracking information to SBS $j$ can be given by:
\begin{equation}\label{eq:DijU}
 D_{ij}^\textrm{U}\left( \boldsymbol{v}_{ij},\sigma_{i}^{\max}\right)=\frac{{{K_{i}\left(\sigma_{i}^{\max}\right)}}}{{{c_{ij}\left(\boldsymbol{v}_{ij}\right)}}},
 \end{equation}
where ${K_{i}\left(\sigma_{i}^{\max}\right)}$ is the size of tracking information data {that user $i$ needs to transmit} and $\sigma _{i}^{\max}=\mathop {\max }\limits_{k\in \mathcal{U}_j, k \ne i} \left( {{\sigma _{ik}}} \right)$ represents {\color{black}the maximum covariance for each user $i$.} From (\ref{eq:DijU}), we can see that finding data correlation among the users allows minimizing the uplink data traffic that the users need to transmit.

  \subsection{Model of Successful Transmission Probability} 
% In order to jointly consider the transmission delay in both uplink and downlink, we introduce a method based on the framework of multi-attribute utility theory \cite{abbas2010constructing} 
Next, we %To construct an appropriate utility function that can capture transmission delay in both uplink and downlink, we
derive the successful transmission probability.
  
 For each user $i$ that requests content $a_{it}$ at time $t$, a successful transmission is defined as:
 \begin{equation}\label{eq:S}
 Q_{it}\!\left(a_{it}, \boldsymbol{s}_{ij},  \boldsymbol{v}_{ij}, g_{ja_{it}},\sigma_{i}^{\max} \right)\!=\!{{\mathbbm{1}_{\!\left\{ {D_{ijt}\left(\boldsymbol{s}_{ij},g_{ja_{it}}\!\right)+ D_{ijt}^\textrm{U}\left( \boldsymbol{v}_{ij}, \sigma_{i}^{\max}\right) \le \gamma _\textrm{D}} \!\right\}}}},
 \end{equation}
 where $D_{ijt}$ and $D_{ijt}^\textrm{U}$ represent the delay of downlink and uplink transmission at time $t$, respectively. $\gamma _\textrm{D}$ is the delay requirement for each VR user. Based on (\ref{eq:S}), the probability of successful transmission can be given by: 
 \begin{equation}\label{eq:Pi}
\begin{split}
{\mathbbm{P}_i}\left( \boldsymbol{s}_{ij}, \boldsymbol{v}_{ij} \right) = \mathop \frac{1}{T}\sum\limits_{t = 1}^T Q_{it}\!\left(a_{it}, \boldsymbol{s}_{ij},  \boldsymbol{v}_{ij}, g_{ja_{it}},\sigma_{i}^{\max} \right),
\end{split}
\end{equation}
{\color{black}where $T$ denotes the number of time slots used to evaluate the successful transmission probability of each user.} 
 %(\ref{eq:Pi}) captures the probability of successful transmission for user $i$. 
 Now, given the resource block vectors $\boldsymbol{s}_{ij}$ and $ \boldsymbol{v}_{ij}$, {users association can be determined. Once the user association is determined, the data correlation of each user is also determined, since the data correlation depends on the user association and the data request by each user. In consequence, {one can consider only the optimization of the resource block allocation for each SBS to maximize the successful transmission probability of each user.} }  
To capture the gain that stems from the allocation of the resource blocks, we state the following result:
 
 \begin{theorem}\label{th:2}
\emph{Given the uplink and downlink resource blocks $\boldsymbol{v}_{ij}$ and $\boldsymbol{s}_{ij}$ as well as the uplink and downlink data correlation $\sigma_i^{\max}$ and $M\left( {g}_{ja}\right) $, 
the gain of user $i$'s successful transmission probability due to an increase in the amount of allocated resource blocks and the change of content transmission format includes:}

\emph{\romannumeral1) The gain due to an increase of the number of uplink resource blocks, $\Delta {\mathbbm{P}_i}$, is given by:
\begin{equation}
 \setlength{\abovedisplayskip}{4 pt}
\setlength{\belowdisplayskip}{4 pt}
\Delta {\mathbbm{P}_i}= \mathop \frac{1}{T}\sum\limits_{t = 1}^T {{\mathbbm{1}_{\left\{   \gamma_\textrm{DU}-c_{ij}\left(\Delta\boldsymbol{v}_{ij}\right) \le {c_{ij}\left(\boldsymbol{v}_{ij}\right) < \gamma_\textrm{DU} }  \right\}}}},
\end{equation}
where {\color{black}$\Delta {\mathbbm{P}_i}$ represents the change of the successful transmission probability of each user $i$} and $\gamma_\textrm{DU}=\frac{ {{{K_{i}\left(\sigma_{i}^{\max}\right)}}}   }{\gamma _D-D_{ijt}\left( \boldsymbol{s}_{ij}, g_{ja_{it}} \right)}$.}
 
\emph{\romannumeral2) The gain due to an increase of the number of downlink resource blocks, $\Delta {\mathbbm{P}_i}$, is given by:
\begin{equation}
\setlength{\abovedisplayskip}{4 pt}
\setlength{\belowdisplayskip}{4 pt}
\Delta {\mathbbm{P}_i}=\mathop \frac{1}{T}\sum\limits_{t = 1}^T {{\mathbbm{1}_{\left\{ \gamma_\textrm{DD}-c_{ij}\left(\Delta\boldsymbol{s}_{ij}\right) \le {c_{ij}\left(\boldsymbol{s}_{ij}\right) < \gamma_\textrm{DD} } \right\}}}}
\end{equation}
where $\gamma_\textrm{DD}=\frac{G_{120^\circ}}{\gamma _D-D_{ijt}^\textrm{U}\left( \boldsymbol{v}_{ij},\sigma_{i}^{\max}\right) -\frac{{{M\left( {g}_{ja} \right)}}}{V_{i}^\textrm{D}} }$.}

\emph{\romannumeral3) The gain due to the change of content transmission format, $g_{ja}$, $\Delta {\mathbbm{P}_i}$, can be given by:
\begin{equation}
\Delta {\mathbbm{P}_i}\!\! =\!\!\left\{ {\begin{array}{*{20}{c}}
  {\!\!\!\frac{1}{T}\sum\limits_{t = 1}^T {{\mathbbm{1}_{\left\{ {   {{{M\left( 120^\circ \right)}}}  \leqslant  \left( \gamma_D- \frac{{{ G_{120^\circ} }}}{{{c_{ij}\left(\boldsymbol{s}_{ij}  \right)}}}    -D_{ijt}^\textrm{U}\left( \boldsymbol{v}_{ij},\sigma_{i}^{\max} \right) \right){V_{i}^\textrm{B}}} < {{{M\left( 360^\circ \right)}}}   \right\}}},~\text{if}~M\left( 360^\circ \right)> M\left( 120^\circ \right)} }, \\ 
  {\!\!\!\frac{1}{T}\sum\limits_{t = 1}^T {{\mathbbm{1}_{\left\{ {   {{{M\left( 360^\circ \right)}}}  \leqslant  \left( \gamma_D- \frac{{{ G_{120^\circ} }}}{{{c_{ij}\left(\boldsymbol{s}_{ij}  \right)}}}    -D_{ijt}^\textrm{U}\left( \boldsymbol{v}_{ij},\sigma_{i}^{\max} \right) \right){V_{i}^\textrm{B}}} < {{{M\left( 120^\circ \right)}}}   \right\}}},~\text{if}~M\left( 360^\circ \right)<M\left( 120^\circ \right)} }, \\ 
  {\;\;\;\;\;\;\;\;\;\;\;\;\;\;\;\;\;\;\;\;\;\;\;\;\;\;\;\;\;\;\;\;\;\;\;0,\;\;\;\;\;\;\;\;\;\;\;\;\;\;\;\;\;\;\;\;\;\;\;\;\;\;\;\;\;\;\;\;\;\;\;\;\;\;\;\;~\text{if}~M\left( 360^\circ \right)=M\left( 120^\circ \right)}. \\
\end{array}} \right.
\end{equation} 
}
\end{theorem}
\begin{proof} See Appendix A.
\end{proof}
Theorem \ref{th:2} shows that the resource block allocation scheme and content transmission format will jointly determine the users' successful transmission probability. Indeed, Theorem \ref{th:2} can provide guidance for action selection in the machine learning approach proposed in Section \ref{se:al}.
Theorem \ref{th:2} also shows that, as $M\left(120^\circ\right)=M\left(360^\circ\right)$, the users' successful transmission probability will remain constant as the content transmission format changes. The reason behind this is that, as the number of users that request the same content increases or the data correlation among the users decreases, the data size of visible contents that the cloud needs to transmit will be equal to or larger than the data size of a $360^\circ$ content.       
   
 \subsection{Problem Formulation}
 Having defined our system model, the next step is to introduce a resource management mechanism to effectively allocate the downlink and uplink resources so as to maximize the successful transmission probability of all users. This problem will be: 
 \addtocounter{equation}{0}
\begin{equation}\label{eq:max}
\begin{split}
\mathop {\max }\limits_{\boldsymbol{s}_{ijn},\boldsymbol{v}_{ijn},{{g}_{jan}}, {\sigma}_i^{\max}}  \sum\limits_{n = 1}^N \sum\limits_{j \in \mathcal{B}}\sum\limits_{i \in \mathcal{U}_j}  {\mathbbm{P}_{in}}\left( \boldsymbol{s}_{ijn}, \boldsymbol{v}_{ijn}\right),
%=\\&\!\!\!\!\mathop {\max }\limits_{\boldsymbol{g}_i,\boldsymbol{f}_j,{\mathcal{U}_j},\mathcal{S}_j} \mathop {\lim }\limits_{T \to \infty } \frac{1}{T}\sum\limits_{t= 1}^T  \sum\limits_{j \in \mathcal{B}} \sum\limits_{i \in \mathcal{U}_j}   {{\mathbbm{1}_{\left\{ {T_{it}^\textrm{P}\left( a_{it}, {g}_{ia}, \boldsymbol{f}_j, U_j, \mathcal{S}_j \right) + T_{it}^\textrm{M}\left( {a_{it},{{g}_{ia}},{\mathcal{S}_j},{\mathcal{U}_j}} \right) \le D} \right\}}}},
\end{split}
\end{equation}
\vspace{-0.0cm}
\begin{align}\label{c1}
&\;\;\;\;\rm{s.\;t.}\scalebox{1}{$\;\;\;\; v_{ijn,k} \in \left\{ 0, 1 \right\},\;\;\;\forall j \in \mathcal{B}, i \in \mathcal{U}_j, $} \tag{\theequation a}\\
%&\scalebox{1}{\small$R_{ki}\!\left(t\right)\!\!\in\!\! \left\{ {{\!\!R_{lki}}}\!\!\left(u_{ki}\!\left(t\right)\!\right)\!,{{R_{uki}}}\!\left(e_{ki}\!\left(t\right)\!\right)\!,\!{\frac{{{R_{lki}\!\left(\!u_{ki}\!\left(t\right)\!\right)}{R_{Ck}\!\left(t\right)}}}{{{R_{lki}\!\left(\!u_{ki}\!\left(t\right)\!\right)} \!+ \!{R_{Ck}\!\left(t\right)}}}}\!,\!{\frac{{{R_{uki}\!\left(\!e_{ki}\!\left(t\right)\!\right)}\!{R_{Ck}\!\left(t\right)\!}}}{{{R_{uki}\!\left(e_{ki}\!\left(t\right)\!\right)} \!+\! {R_{Ck}\!\left(t\right)\!}}}}\!\right\},$} \tag{\theequation b}\\
&\scalebox{1}{$\;\;\;\;\;\;\;\;\;\;\;\;\;\; s_{ijn,k} \in \left\{ 0, 1 \right\},\;\;\;\forall j \in \mathcal{B}, i \in \mathcal{U}_j, $} \tag{\theequation b}\\
&\scalebox{1}{$\;\;\;\;\;\;\;\;\;\;\;\; \;\; g_{jan} \in \left\{ {120^\circ}, {360^\circ} \right\},\;\;\;\forall i \in \mathcal{U}_j, a \in \mathcal{C},$}   \tag{\theequation c}
%&\scalebox{1}{$\;\;\;\;\;\;\;\;\;\;\;\; {{e_{ki}\left(t\right)}}{{u_{ki}\left(t\right)}}=0, R_{ki}\left(t\right)R_{ji}\left(t\right)=0, j \ne k, \forall k,j \in \mathcal{K},$} \tag{\theequation c}\\
%$L = \sum\nolimits_i {{a_i}}  + \sum\nolimits_i {{b_i}}$
%&\scalebox{1}{$\;\;\;\;\;\;\;\;\;\;\;\;{x_{li}} \in \left\{ {0,1} \right\},\;\;\;\forall i \in \mathcal{U},  $} \tag{\theequation c}\\
%&\scalebox{1}{$\;\;\;\;\;\;\;\;\;\;\;\; {R_i\left(t\right)}\!\! \in\!\! \left\{\!\! {{R_{lki}\!\left(t\right)},{R_{uki}\!\left(t\right)},\frac{{{R_{uki}\left(t\right)}{R_{Ck}\left(t\right)}}}{{{R_{uki}\left(t\right)} + {R_{Ck}\left(t\right)}}},\frac{{{R_{lki}\left(t\right)}{R_{Ck}\left(t\right)}}}{{{R_{lki}\left(t\right)} + {R_{Ck}\left(t\right)}}}} \!\right\},  $} \tag{\theequation c}\\ 
\end{align}
where ${\mathbbm{P}_{in}}$ represents the successful transmission probability during a period $n$ that consists of $T$ time slots. $\boldsymbol{s}_{ijn}$ and $\boldsymbol{v}_{ijn}$ represent the resource block allocation during period $n$. Here, the content request distribution of each user will change as period $n$ varies. (\ref{eq:max}a) and (\ref{eq:max}b) indicate that each uplink and downlink resource block $k$ can be only allocated to one user. (\ref{eq:max}c) indicates that the cloud can transmit visible or $360^\circ$ contents to the SBSs. {\color{black}From (\ref{eq:max}), we can see that the successful transmission probability ${\mathbbm{P}_{in}}$ is optimized over $N$ periods. The users' content request distributions and data correlation will change as period $n$ varies. Since the content request distribution of each user changes during each period, each SBS needs to change its resource allocation strategy so as to optimize the successful transmission probability during each period.
Moreover, from (\ref{eq:max}a) to (\ref{eq:max}c), we can see that the optimization variables $v_{ijn,k}$ and $s_{ijn,k}$ are binary, and $g_{jan}$ is discrete. Consequently, we cannot differentiate the optimization function. Thus, the problem in (\ref{eq:max}) cannot be readily solved by conventional optimization algorithms. In addition, from (\ref{eq:sigma}), we can see that user association (the number of users that request different visible content $a$) and data correlation are coupled. Meanwhile, from (\ref{eq:Dsg}) and (\ref{eq:DijU}), we can see that resource allocation and data correlation jointly determine the transmission delay of each user. Hence, in (\ref{eq:max}), the user association, resource allocation, and data correlation are interdependent and we cannot divide the optimization problem into three separate optimization problems. Also, from (\ref{eq:cd}) and (\ref{eq:cv}), we can see that the data rate of each user $i$ depends on not only the resource allocation scheme that is implemented by its associated SBS but also on the resource allocation schemes performed by other SBSs. Finally, since $\small {\mathbbm{P}_i}\left( \boldsymbol{s}_{ij}, \boldsymbol{v}_{ij} \right)=\mathop \frac{1}{T}\sum\limits_{t = 1}^T {{\mathbbm{1}_{\!\left\{ {D_{ijt}\left(\boldsymbol{s}_{ij},g_{ja_{it}}\!\right)+ D_{ijt}^\textrm{U}\left( \boldsymbol{v}_{ij}, \sigma_{i}^{\max}\right) \le \gamma _\textrm{D}} \!\right\}}}}$ and $\small {{\mathbbm{1}_{\!\left\{ {D_{ijt}\left(\boldsymbol{s}_{ij},g_{ja_{it}}\!\right)+ D_{ijt}^\textrm{U}\left( \boldsymbol{v}_{ij}, \sigma_{i}^{\max}\right) \le \gamma _\textrm{D}} \!\right\}}}}$ are non-convex functions, the problem in (\ref{eq:max}) is challenging to solve.} %From (\ref{eq:max}), we can see that the elements in $\boldsymbol{s}_{ijn}$ and $\boldsymbol{v}_{ijn}$ are binary variables. Meanwhile, the content transmission format $g_{jan}$ is discrete. Moreover, the resource allocation, user association, and content transmission are coupled. In addition, since the content request distribution of each user changes during each period, each SBS needs to change its resource allocation strategy so as to optimize the successful transmission probability during each period. In consequence, the problem in (\ref{eq:max}) cannot be solved by conventional optimization algorithms \cite{8374947,8382257,hosseinalipour2018two}. Therefore, we propose a transfer reinforcement learning (RL) algorithm to solve it since a transfer RL algorithms \cite{7295483} can learn the difference of the successful transmission probability as the network and users states change so as to adaptively adjust the resource block allocation policy based on users' content request distribution and data correlations.      

 \section{Echo State Networks for Self-Organizing Resource Allocation} \label{se:al}
In this section, a transfer reinforcement learning (RL) algorithm based on the neural network framework of ESNs \cite{26,chen2016caching,APractical} is introduced. The proposed transfer RL algorithm can be used to find the optimal resource block allocation during each period so as to maximize the users' successful transmission probability. Conventional learning approaches such as Q-learning usually use a matrix to record the information related to the users and networks. As a result, the information that the Q-learning approach must record will exponentially increase, when the number of SBSs and users in the network increases. {In consequence, the Q-learning approach cannot record all of the information related to the users and network. However, the ESN-based transfer RL algorithm exploits a function approximation method to record all of the information related to the network and users. Hence, the proposed ESN-based transfer RL  algorithm can be used for large networks with dense users.} Moreover, the users' content request distributions and data correlation will change as time elapses. Traditional learning approaches such as \cite{chu2018reinforcement} must {\color{black}re-implement the learning process} as the users' content request distribution and data correlation change. However, the  ESN-based transfer RL algorithm can transform the already learned resource block allocation policy into the new resource block allocation policy that must be learned as the users' content request distribution and data correlation change so as to improve the convergence speed. %After learning this relationship, the proposed algorithm can use the historic learning result to find the optimal resource allocation policy for each SBS. 

%The proposed transfer RL algorithm consists of two components: (i) ESN-based RL algorithm and (ii) ESN-based transfer learning algorithm. 
%Next, we first introduce the components of the ESN-based transfer RL algorithm. Then, we specify the process of using ESN-based transfer RL algorithm to solve the problem in (\ref{eq:max}).
\subsection{Components of ESN-based Transfer RL Algorithm}   

 \begin{figure}[!t]
  \begin{center}
   \vspace{0cm}
    \includegraphics[width=9cm]{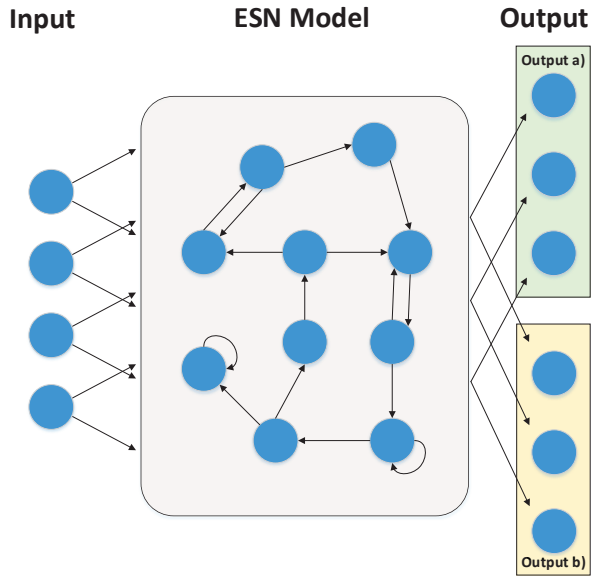}
    \vspace{-0.2cm}
    \caption{\label{algorithm} The components of the ESN-based transfer RL algorithm. Here, output a) combined with the input and the ESN model is used as an RL algorithm so as to find the optimal resource allocation scheme while output b) combined with the input and the ESN model acts as a transfer learning algorithm to transfer the alreadly learned information to new environments.}
  \end{center}\vspace{-0.85cm}
\end{figure}

The ESN-based transfer RL algorithm of each SBS $j$ consists of four components: (a) input, (b) action, (c) output, and d) ESN model, as shown in Fig. \ref{algorithm}, which can be given by:
\begin{itemize}
\item \emph{Input:} The input of the ESN-based transfer RL algorithm the strategy index of the SBSs and period $n$, which is $\boldsymbol{x}_{j}\left(t\right)=\left[ {\pi_{1}\left(t\right), \cdots , \pi_{B}\left(t\right), n} \right]^{\mathrm{T}}$ where $\pi_{k}\left(t\right)$ is the index of a strategy that SBS $k$ uses at time $t$. Here, the strategies of each SBS are determined by the $\epsilon$-greedy mechanism \cite{30}. 
 \item \emph{Action:} The action $\boldsymbol{a}_j$ of each SBS $j$ consists of downlink resource allocation vector $\boldsymbol{s}_j=\left[\boldsymbol{s}_{1j},\boldsymbol{s}_{2j}, \dots , \boldsymbol{s}_{{U}_jj}\right]$ and uplink resource allocation vector $\boldsymbol{v}_j=\left[\boldsymbol{v}_{1j},\boldsymbol{v}_{2j}, \dots , \boldsymbol{v}_{{U}_jj}\right]$. $U_j$ represents the number of VR users located within the coverage of SBS $j$.   
 %Each SBS implements one ELSM algorithm for maximizing the reliability of its associated users.
 \item \emph{Output:} The output of the ESN-based transfer RL algorithm consists of two components: a) predicted successful transmission probability and b) predicted variation in the successful transmission probability when the users' content request distribution and data correlation change. Output a) is used to find the relationship between the strategies $\boldsymbol{\pi}_j$, actions $\boldsymbol{a}_j$, and users' successful transmission probability $\sum\limits_{i \in \mathcal{U}_j}  {\!\mathbbm{P}_{in}}\!\left( \boldsymbol{a}_{j}\left(t\right)\right)$. Therefore, output a) of the ESN-based transfer RL algorithm can be given by $\boldsymbol{y}_{j}\left(t\right)=\left[ {y_{j\boldsymbol{a}_{j1}}\left(t\right), \cdots , y_{j\boldsymbol{a}_{jN_{ja}}}\left(t\right)} \right]^{\mathrm{T}}$. Here, $ y_{j\boldsymbol{a}_{jn}}\left(t\right)$ represents the predicted total successful transmission probability of SBS $j$ using action $\boldsymbol{a}_{jn}$. $N_{ja}$ is the total number of actions of each SBS $j$.
 
 {\color{black}To calculate $\sum\limits_{i \in \mathcal{U}_j}  {\!\mathbbm{P}_{in}}\!\left( \boldsymbol{a}_{j}\left(t\right)\right)$, we need to determine the transmission format of each content. Given action $\boldsymbol{a}_{j}\left(t\right)$, the user association will be determined. In consequence, for the cloud, the choice of visible or $360^\circ$ content transmission can be given by the following theorem:

\begin{theorem}\label{th:1}
\emph{Given action $ \boldsymbol{a}_{j}\left(t\right)$, the maximum downlink backhaul link rate $V^\textrm{B}$, the set $\mathcal{U}_{ja}$ of ${U}_{ja}$ users that request different visible contents extracted from $360^\circ$ content $a$, and the set of data correlation $\mathcal{C}_{a}=\bigcup\limits_{n = 2}^{{U_{ja}}} {\mathcal{C}_a^n} $, the transmission format of content $a$ can be given by:
\begin{itemize}
\item If $G_{360^\circ} \geqslant  L_a\left(\mathcal{C}_{a}\right)$, $g_{ja}=120^\circ$.
\item If $G_{360^\circ} < L_a\left(\mathcal{C}_{a}\right)$, $g_{ja}=360^\circ$.
\end{itemize}
Here, $L_a\left(\mathcal{C}_{a}\right)= G_{120^\circ} \left( {{U_{ja}} - \sum\limits_{n = 2}^{{U_{ja}}} {\sum\limits_{{C_a} \in \mathcal{C}_a^n} {{{\left( { - 1} \right)}^{n - 1}}{C_a}} }    } \right)$.
}

\end{theorem}
\begin{proof} See Appendix B.
\end{proof}
 Theorem \ref{th:1} shows that the choice of $360^\circ$ and visible content transmission depends on the data size of the $360^\circ$ and visible contents, the data correlation among the users, and the number of users that request the same content. When the data correlation among the users increases, the cloud prefers to transmit visible contents to the SBSs. In contrast, as the number of users that request the same content decreases, the cloud prefers to transmit the $360^\circ$ content. Based on Theorem 2 and action $\boldsymbol{a}_{j}\left(t\right)$, $\sum\limits_{i \in \mathcal{U}_j}  {\!\mathbbm{P}_{in}}\!\left( \boldsymbol{a}_{j}\left(t\right)\right)$ can be computed.}

The output b) is used to find the relationship between $ {\!\mathbbm{P}_{in}}\!\left( \boldsymbol{a}_{j}\left(t\right)\right)$ and $ {\!\mathbbm{P}_{in+1}}\!\left( \boldsymbol{a}_{j}\left(t\right)\right)$ when SBS $j$ only knows $ {\!\mathbbm{P}_{in}}\!\left( \boldsymbol{a}_{j}\left(t\right)\right)$. This means that the proposed algorithm can transfer the information from the already learned successful transmission probability $ {\!\mathbbm{P}_{in}}\!\left( \boldsymbol{a}_{j}\left(t\right)\right)$ to the new successful transmission probability $ {\!\mathbbm{P}_{in+1}}\!\left( \boldsymbol{a}_{j}\left(t\right)\right)$ that must be learned. 
  The output b) of the ESN-based transfer learning algorithm at time $t$ is the predicted variation in the successful transmission probability when the users' information changes $\boldsymbol{y}'_{j}\left(t\right)=\left[ {y'_{j\boldsymbol{a}_{j1}}\left(t\right), \cdots , y'_{j\boldsymbol{a}_{jN_{ja}}}\left(t\right)} \right]^{\mathrm{T}}$ with 
  ${y}'_{j\boldsymbol{a}_{jk}}\left(t\right)= {\!\mathbbm{P}_{in}}\!\left( \boldsymbol{a}_{jk}\left(t\right)\right)- {\!\mathbbm{P}_{in-1}}\!\left( \boldsymbol{a}_{jk}\left(t\right)\right)$.   
 
 %The output of the ESN-based RL algorithm is the predicted total successful transmission probability of each SBS $j$, which can be given by $\boldsymbol{y}_{j}\left(t\right)=\left[ {y_{j\boldsymbol{a}_{j1}}\left(t\right), \cdots , y_{j\boldsymbol{a}_{jN_{ja}}}\left(t\right)} \right]^{\mathrm{T}}$. Here, $ y_{j\boldsymbol{a}_{jn}}\left(t\right)$ represents the predicted successful transmission probability of SBS $j$ using action $\boldsymbol{a}_{jn}$ and $N_{ja}$ is the total number of actions of SBS $j$.
 
  \item \emph{ESN Model:} An ESN model of the learning approach can approximate the function of the ESN input $\boldsymbol{x}_{j}\left(t\right)$ and output $\boldsymbol{y}_{j}\left(t \right)$ as well as output $\boldsymbol{y}'_{j}\left(t \right)$. The ESN model is composed of two output weight matrices $\boldsymbol{W}_j^{\textrm{out}}, \boldsymbol{W}_j^{'\textrm{out}} \in {\mathbb{R}^{N_{ja} \times \left(N_w+B+1\right)}}$. Let $\boldsymbol{W}_j^{\textrm{in}} \in {\mathbb{R}^{N_w \times \left(B+1\right)}}$ be the input weight matrix and $\boldsymbol{W}_j \in \mathbb{R}^{N_w \times N_w}$ be the recurrent matrix with $N_w$ being the number of the neurons. Mathematically, $\boldsymbol{W}_j$ can be given by:
 \begin{equation}
\boldsymbol{W}_l=\left[ {\begin{array}{*{20}{c}}
{{0}}&{{0}}& \cdots &{{w}}\\
{{w}}&0&0&0\\
0& \ddots &0&0\\
0&0&{{w}}&0
\end{array}} \right],
\end{equation}
where $w \in \left[0,1\right]$ is a constant. Here, the recurrent weight matrix $\boldsymbol{W}_j$ combined with the output weight matrices can store the historical information of ESN. This information that includes ESN input, neuron states, and output can be used to find the relationship between the ESN input and output. 
  \end{itemize}
  \subsection{ESN-based Transfer RL for Resource Allocation}
  Next, we introduce the process that uses the ESN-based transfer RL algorithm to solve the problem in (\ref{eq:max}). At each time slot $t$, each SBS $j$ will broadcast its strategy  
  to other SBSs. Then, each SBS can set the input of the ESN-based transfer algorithm. Given the input ${\boldsymbol{x}_{j}\left(t\right)}$, each ESN needs to update the states of the neurons located in the ESN model. The states of the neurons will be given by: 
 \begin{equation}\label{eq:state2}
{\boldsymbol{\mu}_{j}\left(t\right)} ={\mathop{f}\nolimits}\!\left( {\boldsymbol{W}_j{\boldsymbol{\mu}_{j}\left(t-1\right)} + \boldsymbol{W}_j^{\textrm{in}}{\boldsymbol{x}_{j}\left(t\right)}} \right),
\end{equation}
where ${\boldsymbol{\mu}_{j}\left(t-1\right)}$ is the neuron state vector at time slot $t-1$ and $f\!\left(x\right)=\frac{{{e^x} - {e^{ - x}}}}{{{e^x} + {e^{ - x}}}}$ is the tanh function. From (\ref{eq:state2}), we can see that the states of the neurons depend not only on the recurrent input but also on the historical states. In consequence, the ESN model can record historical information related to the inputs, states, and outputs of each ESN. Based on the states of the neurons, the ESN-based transfer RL algorithm will combine with the output weight matrix $\boldsymbol{W}_j^\textrm{out}$ to predict the successful transmission probability of each SBS, which can be given by:
\begin{equation}\label{eq:update}
\boldsymbol{y}_{j}\left(t\right) = {\boldsymbol{W}_{j}^{\textrm{out}}\left(t\right)}  \left[ {\begin{array}{*{20}{c}}
  { {{\boldsymbol{\mu}_{j}\left(t\right)}}} \\ 
  {\boldsymbol{x}_{j}\left(t\right)} 
\end{array}} \right],
\end{equation}
where ${\boldsymbol{W}_{j}^{\textrm{out}} \left(t\right)}$ is the output weight matrix at time slot $t$. Meanwhile, the proposed transfer RL approach  combined with ${\boldsymbol{W}_{j}^{'\textrm{out}} \left(t\right)}$ will predict the variation in the successful transmission probability when the users' content request distribution and data correlation change. This prediction process can be given by: 
\begin{equation}\label{eq:update1}
\boldsymbol{y}'_{j}\left(t\right) = {\boldsymbol{W}_{j}^{'\textrm{out}}\left(t\right)}  \left[ {\begin{array}{*{20}{c}}
  { {{\boldsymbol{\mu}_{j}\left(t\right)}}} \\ 
  {\boldsymbol{x}_{j}\left(t\right)} 
\end{array}} \right].
\end{equation}
From (\ref{eq:update}) and (\ref{eq:update1}), we can see that, to enable an ESN to predict different outputs (i.e., $\boldsymbol{y}'_{j}\left(t\right) $ or $\boldsymbol{y}_{j}\left(t\right) $), we only need to adjust the output weight matrix of a given ESN. The adjustment of the output weight matrix ${\boldsymbol{W}_{j}^{\textrm{out}}}$ can be given by: 
\begin{equation}\label{eq:w2}
\begin{split}
&{\boldsymbol{W}_{jk}^{\textrm{out}}\left(t + 1\right)} = {\boldsymbol{W}_{jk}^{\textrm{out}}\!\left(t\right)} \!+\! {\lambda}\! \left( { \sum\limits_{i \in \mathcal{U}_j}  {\!\mathbbm{P}_{in}}\!\left( \boldsymbol{a}_{jk}\left(t\right)\right)-y_{ja_{jk}\left(t\right)}\left(t\right)}\! \right)\!{{\boldsymbol{\mu}}_{j}^{\mathrm{T}}\left(t\right)},
\end{split}
\end{equation}
where $\lambda$ is the learning rate, ${\boldsymbol{W}_{jk}^{\textrm{out}}\left(t + 1\right)} $ is row $k$ of the output weight matrix ${\boldsymbol{W}_{j}^{\textrm{out}}\left(t + 1\right)}$, and $\sum\limits_{i \in \mathcal{U}_j}  {\!\mathbbm{P}_{in}}\!\left( \boldsymbol{a}_{jk}\left(t\right)\right)$ is the actual successful transmission probability resulting from SBS $j$ using action $\boldsymbol{a}_{jk}\left(t\right)$. Similarly, ${\boldsymbol{W}_{j}^{'\textrm{out}}}$ can be adjusted based on the following equation: 
\begin{equation}\label{eq:w3}
\begin{split}
&{\boldsymbol{W}_{jk}^{'\textrm{out}}\left(t + 1\right)} ={\boldsymbol{W}_{jk}^{'\textrm{out}}\!\left(t\right)} + {\lambda'}\! \left( { \sum\limits_{i \in \mathcal{U}_j} {\!\mathbbm{P}_{in}}\!\left( \boldsymbol{a}_{jk}\left(t\right)\right)\!-\!\sum\limits_{i \in \mathcal{U}_j}  {\!\mathbbm{P}_{in-1}}\!\left( \boldsymbol{a}_{jk}\left(t\right)\right)-y_{ja_{jk}\left(t\right)}\left(t\right)}\! \right)\!{{\boldsymbol{\mu}}_{j}^{\mathrm{T}}\left(t\right)},
\end{split}
\end{equation}
where ${\lambda'}$ is the learning rate (${\lambda'}  \ll  {\lambda} $). Based on (\ref{eq:state2})-(\ref{eq:w3}), the ESN-based transfer RL algorithm can predict: a) the successful transmission probability resulting from each action that SBS $j$ takes and b) the variation in the successful transmission probability when the users' information changes. In consequence, each SBS $j$ will first use the output a) of the ESN-based transfer RL algorithm to find the optimal resource allocation scheme so as to maximize the users' successful transmission probability. Then, as the users' content request distribution or data correlation change, each SBS $j$ can use the outputs b) to find the relationship between the already learned successful transmission probability and the new successful transmission probability that must be learned. In consequence, each SBS $j$ can directly transfer the already learned successful transmission probability to the new successful transmission probability so as to increase the convergence speed.   
The proposed approach that is implemented by each SBS $j$ is summarized in Table~I. 

\subsection{Complexity and Convergence}
{\color{black}
With regards to the computational complexity, the complexity of the proposed algorithm depends on the action that is performed at each iteration. The proposed algorithm is used to find the optimal action. As the number of iterations needed to find and perform the optimal action increases, the complexity of the proposed algorithm increases. However, the action selection depends on the $\epsilon$-greedy mechanism which will also change with time. Therefore, for a very general case, we cannot quantitatively analyze the complexity of the proposed algorithm. Hence, we can only analyze the worst-case complexity of the proposed algorithm.       
Since the worst-case for each SBS is to traverse all actions, the worst-case complexity of the proposed algorithm is $O(\left| {{\mathcal{A}_{1}}}\right| \times \dots \times \left| {{\mathcal{A}_{K}}}\right|)$ where $\left| {{\mathcal{A}_{j}}}\right|$ denotes the total number of actions of each SBS $j$. However, the worst-case complexity pertains to a rather unlikely scenario in which all SBSs choose their optimal resource allocation schemes after traversing all other resource allocation schemes during each period $n$. Moreover, the proposed algorithm uses a function approximation method to find the relationship between the actions, states, and utilities. In this context, the proposed algorithm will not need to traverse all actions to find this relationship. In addition, ESNs are a type of recurrent neural networks which can use historical input data to find the relationship between actions, states, and utilities which can reduce the training complexity.   
Furthermore, unlike the existing learning algorithms such as long short term memory based RL algorithms \cite{8359094} that need to calculate the gradients of all neurons in the hidden and input layers, the proposed algorithm only need to update the output weight matrix. 
Moreover, at each iteration, the proposed transfer learning algorithm only needs to update one row of each output weight matrix. In particular, since $\boldsymbol{W}_j^{\textrm{out}}, \boldsymbol{W}_j^{'\textrm{out}} \in {\mathbb{R}^{N_{ja} \times \left(N_w+B+1\right)}}$, the proposed algorithm must update $\boldsymbol{W}_{jk}^{\textrm{out}}, \boldsymbol{W}_{jk}^{'\textrm{out}} \in {\mathbb{R}^{1 \times \left(N_w+B+1\right)}}$. This will also significantly reduce the training complexity of our algorithm. Finally, compared to the existing RL algorithms, the proposed transfer RL algorithm can transform the already learned resource allocation policy into the new resource allocation policy that must be learned so as to reduce the number of iterations needed for training. 
%As shown in Fig. 10, as period $n$ increases, the number of iterations needed for training decreases and, hence, the training complexity of the proposed algorithm decreases. 
}

{For the convergence of the ESN-based transfer RL approach, we can directly use the result of \cite[Theorem 2]{VROWNchen} which showed that, for each action $\boldsymbol{a}_{jk}\left(t\right)$, the outputs of a given ESN will converge to ${\!\mathbbm{P}_{in}}\!\left( \boldsymbol{a}_{jk}\left(t\right)\right)$ and ${\!\mathbbm{P}_{in}}\!\left( \boldsymbol{a}_{jk}\left(t\right)\right)- {\!\mathbbm{P}_{in-1}}\!\left( \boldsymbol{a}_{jk}\left(t\right)\right)$ via adjusting the values of the input and output weight matrices, as follows.

{\begin{corollary}[follows from \cite{VROWNchen}]\label{co:1}\emph{ The ESN-based transfer learning algorithm of each SBS $j$ converges to the utility values ${\!\mathbbm{P}_{in}}\!\left( \boldsymbol{a}_{jk}\left(t\right)\right)$ and ${\!\mathbbm{P}_{in}}\!\left( \boldsymbol{a}_{jk}\left(t\right)\right)- {\!\mathbbm{P}_{in-1}}\!\left( \boldsymbol{a}_{jk}\left(t\right)\right)$, if any following conditions is satisfied:
\begin{itemize}
\item[\romannumeral1)] $\lambda$ and $\lambda'$ are constant, and $\mathop {\min }\limits_{\boldsymbol{W}_{ji}^\textrm{in},{\boldsymbol{x}_{\tau ,j}},{\boldsymbol{x}'_{\tau ,j}}} \boldsymbol{W}_{ji}^\textrm{in}\left({\boldsymbol{x}_{\tau,j}}-{\boldsymbol{x}'_{\tau',j}}\right) \ge 2$, where $\boldsymbol{W}_{ji}^\textrm{in}$ represents the row $i$ of $\boldsymbol{W}_{j}^\textrm{in}$.
\item[\romannumeral2)] $\lambda$ and $\lambda'$  satisfy the Robbins-Monro conditions {\color{black}(${\lambda\left(t\right)} > 0,\sum\nolimits_{t = 0}^\infty  {{\lambda\left(t\right)} = +\infty ,\sum\nolimits_{t = 0}^\infty  {\lambda^2\left(t\right) <+ \infty } } $ where $\lambda\left(t\right)$ is the learning rate at time $t$.)} \cite{26}.
\end{itemize}
}
\end{corollary}
{\color{black}Since the convergence of the ESN-based algorithm depends on the values of the learning rates, input weight matrix, and recurrent weight matrix, the proof in [7] will still hold for the proposed algorithm which converges to the utility values ${\!\mathbbm{P}_{in}}\!\left( \boldsymbol{a}_{jk}\left(t\right)\right)$ and ${\!\mathbbm{P}_{in}}\!\left( \boldsymbol{a}_{jk}\left(t\right)\right)- {\!\mathbbm{P}_{in-1}}\!\left( \boldsymbol{a}_{jk}\left(t\right)\right)$.}  
Based on Corollary \ref{co:1}, the proposed transfer learning algorithm can adjust the value of input weight matrix of the ESN and the values of learning rates $\lambda$ and $\lambda'$ to guarantee the convergence of the proposed algorithm. }

\begin{table}[!t]\label{tb1}
%\scriptsize
  \centering
  \caption{%\mycaption{%\vspace*{-1em}
    \vspace*{-0.3em} ESN-based transfer RL Algorithm for resource Allocation}\vspace*{-1em}
    \begin{tabular}{p{5.5in}}
      \hline \vspace*{-0.8em}
     \hspace*{0em}\begin{itemize}\vspace*{-0.1em}
\item[] \hspace*{0em} \textbf{for} each time $t$ \textbf{do}.
\item[] \hspace*{1em}(a) Each SBS $j$ predicts ${\!\mathbbm{P}_{in}}$ based on (\ref{eq:update}).
\item[] \hspace*{0.5em} \textbf{if} $t=1$ 
\item[] \hspace*{2em}(b) Set the policy of the action selection $\boldsymbol{\pi}_{j}\left(1\right)$ uniformly.
\item[] \hspace*{0.5em} \textbf{else}
\item[] \hspace*{2em}(c) Set $\boldsymbol{\pi}_{j}\left(t\right)$ based on the $\epsilon$-greedy mechanism.
\item[] \hspace*{0.5em} \textbf{end if}
\item[] \hspace*{1em}(d) Broadcast the action selection policy index to other SBSs.
\item[] \hspace*{1em}(e) Receive the action selection policy index as ESN input $\boldsymbol{x}_{j}\left(t\right)$.
\item[] \hspace*{1em}(f) Perform an action based on the policy of action selection $\boldsymbol{\pi}_{j}\left(t\right)$.
\item[] \hspace*{1em}(g) Calculate the actual successful transmission probability.
\item[] \hspace*{1em}(h) Update the states of the neurons based on (\ref{eq:state2}).
\item[] \hspace*{1em}(i) Adjust $\boldsymbol{W}_{j}^{\textrm{out}}$ based on (\ref{eq:w2}).

\item[] \hspace*{0.5em} \textbf{if} $n>1$
\item[] \hspace*{2 em}(j) Estimate the value of ${\!\mathbbm{P}_{in+1}}- {\!\mathbbm{P}_{in}}$ based on (\ref{eq:update1}).
\item[] \hspace*{2em}(k) Calculate the actual value of ${\!\mathbbm{P}_{in+1}}- {\!\mathbbm{P}_{in}}$.
\item[] \hspace*{2em}(l) Adjust $\boldsymbol{W}_{j}^{'\textrm{out}}$ based on (\ref{eq:w3}).
\item[] \hspace*{0.5em} \textbf{end if}

%\item[] \hspace*{0.5em} \textbf{for} each time $t$ \textbf{do}.
%\item[] \hspace*{1.5em}(g) Perform an action based on $\boldsymbol{p}_{j\left(\tau\right)}$ and calculate the actual
%\item[] \hspace*{2.6em} reliability of the users. 
%%$\hat \boldsymbol{u}_{\tau,j}$.\vspace*{-0.1em}
%\item[] \hspace*{1.5em}(h) Update the states of the reservoir neurons based on (\ref{eq:state}).
%\item[] \hspace*{0.5em} \textbf{end for}
\item[] \hspace*{0em} \textbf{end for}
\end{itemize}\vspace*{-0.3cm}\\
   \hline
    \end{tabular}\label{tab:algo}\vspace{-0.4cm}
\end{table}

\section{Simulation Results}
For simulations, a cellular network deployed within a circular area with radius $r = 500$ m is considered. In this network, $K=5$ SBSs and $U=25$ VR users are uniformly distributed. The bandwidth $B$ of each resource block is set to $10 \times 180$ kHz. %The rate requirement of VR transmission is 25.32 Mbit/s \cite{VROWNchen}. %The tracking inaccuracy, $K_i\left(\boldsymbol{s}_{ij}^\textrm{u}\right)=\mathop{f}\left(0.4 c_{ij}^\textrm{u}\left( \boldsymbol{s}_{ij}^\textrm{u}\right)\right)$ where $c_{ij}^\textrm{u}$ is the uplink rate of user $i$ associated with SBS $j$ and $f\!\left(x\right)=\frac{{{e^x} - {e^{ - x}}}}{{{e^x} + {e^{ - x}}}}$.
%The detailed parameters are listed in . 
We use typical wireless network parameters such as in \cite{8114362,3gpp.36.814, Mozaffari2016Unmanned}, as listed in Table \ref{ta:simulation}. {\color{black}All of the simulation data related to VR is collected from wired HTC Vive VR devices \cite{htc}. We use 5 VR videos and 5 VR games as the total number of contents that can be provided by the cloud. Each user will request its visible contents according to its head movement. The tracking information of each user is extracted from the sensors of the HTC Vive VR devices.}
 For comparison purposes, we use two baselines: 
\begin{itemize}
 \item {\color{black}The first baseline is the Q-learning algorithm in \cite{30} with data correlation, which we refer to as ``Q-learning with data correlation''. For this Q-learning algorithm, the state is set to the ESN's input ${\boldsymbol{x}_j}$, the actions of Q-learning  are the actions defined in our ESN algorithm, and the reward function $r\left( {\boldsymbol{x}_j},{\boldsymbol{a}_j} \right)$ is the total successful transmission probability in (13). At each iteration, this Q-learning algorithm will select an action based on the $\epsilon$-greedy mechanism and, then, use a Q-table to record the states, actions, and the successful transmission probabilities resulting from the actions that the SBSs have implemented. Finally, each SBS will update its Q-table by the following equation: ${Q_t}\left( {{\boldsymbol{x}_j},{\boldsymbol{a}_j}} \right) = \left( {1 - \zeta } \right){Q_{t - 1}}\left( {\boldsymbol{x}_j},{\boldsymbol{a}_j} \right) + \zeta r\left( {\boldsymbol{x}_j},{\boldsymbol{a}_j} \right)$, where $\zeta$ is the learning rate.}
 
%  The Q-learning algorithm in \cite{30} with data correlation. For this Q-learning algorithm, the state is set to the ESN's input ${\boldsymbol{x}_j}$, the actions of Q-learning  are the actions defined in our ESN algorithm, and the reward function $r\left( {\boldsymbol{x}_j},{\boldsymbol{a}_j} \right)$ is the total successful transmission probability in (\ref{eq:max}). 
  
  \item {\color{black} The second baseline is the Q-learning algorithm in \cite{30} without data correlation, which we refer to as ``Q-learning without data correlation''. The settings of the Q-learning without data correlation are similar to the Q-learning algorithm with data correlation. However, in the Q-learning without data correlation algorithm, the cloud will directly transmit $360^\circ$ contents to the SBSs and the users will directly transmit their tracking information to their associated SBSs without the consideration of data correlation.  }
  
    \item {\color{black} The third baseline is the ESN-based transfer RL without data correlation, which we refer to as ``ESN-based RL without data correlation''. The setting of this algorithm is similar to the proposed algorithm. However, this algorithm does not consider the data correlation for downlink VR content transmission and uplink tracking information transmission.}
  %The Q-learning algorithm in \cite{30} without data correlation. For this Q-learning algorithm, the states, actions, and reward functions are similar to the Q-learning algorithm with data correlation. However, for the Q-learning algorithm without data correlation, the cloud will directly transmit $360^\circ$ contents to the SBSs and the users will directly transmit their tracking information to their associated SBSs without the consideration of data correlation.  

\end{itemize}
All statistical results are averaged over a large number of independent runs.  {\color{black}In simulation figures, total successful transmission probability indicates the total successful transmission probability of all the users that are associated with the SBS.}

 %Here, the data of users' localizations are measured from an actual wired Oculus VR devices and the wireless transmission is simulated, in order to compute the tracking accuracy. All statistical results are averaged over a large number of independent runs.

%\begin{table}
%  \newcommand{\tabincell}[2]{\begin{tabular}{@{}#1@{}}#2\end{tabular}}
%\renewcommand\arraystretch{1}
% \caption{
%    \vspace*{-0.05em} SYSTEM PARAMETERS}\vspace*{-0.6em}
%\centering  
%\begin{tabular}{|c|c|c|c|}% ±íÊŸž÷ÁÐÔªËØ¶ÔÆë·œÊœ£¬left-l,right-r,center-c
%\hline
%\textbf{Parameter} & \textbf{Value} & \textbf{Parameter} & \textbf{Value} \\
%\hline
%$F$ & 1000 & $P_B$ & 20 dBm\\
%\hline
%$B$ & 2 MHz & $S$, $V$ & 5, 5\\
%\hline
%$N_w$ & 1000 & $\sigma ^2$ & -95 dBm\\
%\hline
%$ N_{v}$ &6&$\lambda$, $\lambda'$ & 0.03, 0.3 \\
%\hline
%$m$&5&$r_B$& $30$ m\\
%\hline
%$\alpha$ & 2 & $V_{F}$ & 100 Gbit/s\\
%%\hline
%%$L$ & 1 Mbit & $P_{\max}$ & 20 W&$\delta_{S_i,n}$ & 5 Mbit/s\\
%%\hline
%%$\mu_\textrm{LoS}$,$\mu_\textrm{NLoS}$& 2, 2.4 &$N_x, N_s$ &4, 12&$\chi  $ & 15\\
%%\hline
%  %& &$\zeta_1,\zeta_2$&0.5,0.5\\ 
%%\hline
% %$\chi _{\sigma_\textrm{NLoS}}$ & 5.27 & $\beta, \eta $ & 2, 100&$X,Y$& 11.9, 0.13\\ 
%\hline
%\end{tabular}
% \vspace{-0.3cm}
%\end{table}  

\begin{table}\label{ta:simulation}\footnotesize
  \newcommand{\tabincell}[2]{\begin{tabular}{@{}#1@{}}#2\end{tabular}}
\renewcommand\arraystretch{0.8}
 \caption{
    \vspace*{-0.05em}SYSTEM PARAMETERS}\label{ta:simulation}\vspace*{-0.4em}
\centering  
\begin{tabular}{|c|c|c|c|}% ±íÊŸž÷ÁÐÔªËØ¶ÔÆë·œÊœ£¬left-l,right-r,center-c
\hline
\textbf{Parameters} & \textbf{Values} & \textbf{Parameters} & \textbf{Values}\\
\hline
%$\mu_\textrm{LoS}$&2  &$P_U$& 20 dBm&$G_{120^\circ}$ & 12.5 Mbits \\
%\hline
% $\beta$ & 2 & $R_\textrm{S}$& 2 Gbits&  $B_\textrm{VD}$&2 GHz \\
%\hline
% $\mu_\textrm{NLoS}$& 2.4 & $X$ &  11.9 &$P _{B}$ & 30 dBm\\
%\hline
%$ \eta $ &100& $d_0$ & 5 m&$\sigma^2$ & -105 dBm\\
%\hline
$N_W$&100& $\lambda$ &  0.3 \\
\hline
$T$&1000  & $N$ & 100 \\
\hline
$\chi _{\sigma_\textrm{LoS}}$ & 5.3&$P_U$ & 20 dBm  \\
\hline
 $G_{360^\circ}$& 50 Mbits & $P_\textrm{B}$ &  30 dBm   \\ 
\hline
 $S$ & 5 & $V$ & 5   \\ 
\hline
$\kappa$&5 &$ V^\textrm{B} $ &10 Gbits/s \\
\hline
$\lambda'$ & 0.03 &$\gamma_D$ & 20 ms\\
\hline
$\upsilon ^2$& -105 dBm&$G_{120^\circ}$& 12 Mbits \\
\hline
  $\alpha$ &2 &$N_{w}$ & 100\\
  \hline
\end{tabular}
 \vspace{-0.3cm}
\end{table}

 \begin{figure}[!t]
  \begin{center}
   \vspace{0cm}
    \includegraphics[width=11cm]{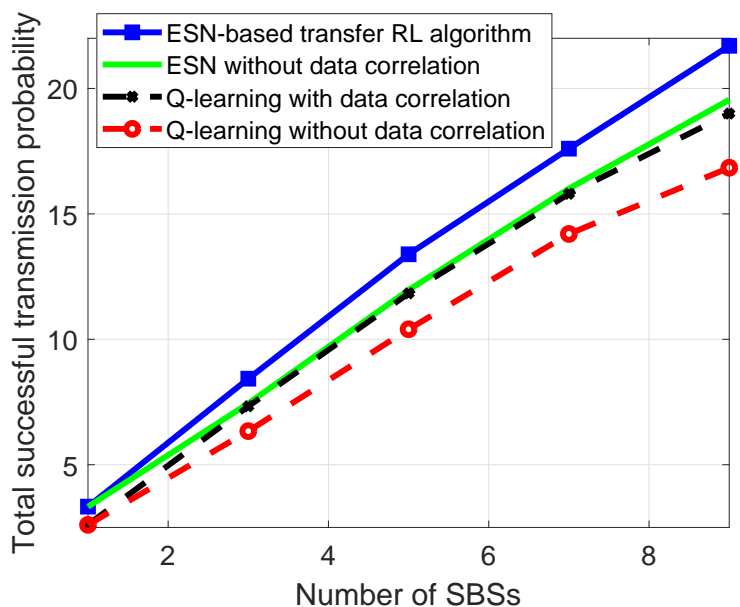}
    \vspace{-0.3cm}
    \caption{\label{figure1} Total successful transmission probability as the number of SBSs varies.}
  \end{center}\vspace{-1cm}
\end{figure}

Fig. \ref{figure1} shows how the total successful transmission probability varies as the number of SBSs changes. Fig. \ref{figure1} shows that, as the number of SBSs increases, the total successful transmission probability of all considered algorithms increases. The reason behind this is that the users have more SBS choices and the number of users located in each SBS's coverage decreases when the number of SBSs increases. Fig. \ref{figure1} also shows that the ESN-based transfer RL algorithm can yield up to 15.8\% and 29.4\% gains in terms of the total successful transmission probability compared to the Q-learning with data correlation and Q-learning without data correlation for a network with 9 SBSs. This is because the ESN-based transfer RL approach can record historical information related to the users' data correlation and content request distribution so as to find the optimal resource block allocation policy.

 \begin{figure}[!t]
  \begin{center}
   \vspace{0cm}
    \includegraphics[width=11cm]{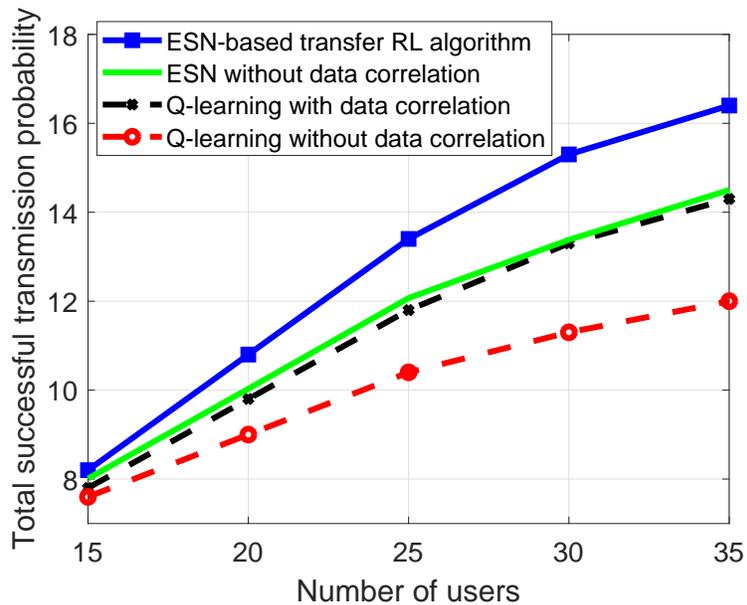}
    \vspace{-0.3cm}
    \caption{\label{figure2} Total successful transmission probability as the number of users varies.}
  \end{center}\vspace{-1cm}
\end{figure}

In Fig. \ref{figure2}, we show how the total successful transmission probability changes as the number of users varies. Fig. \ref{figure2} shows that, as the number of users increases, the total successful transmission probability of all considered algorithms increases. This implies that, as the number of users increases, the SBSs have more choices of users to service. Fig. \ref{figure2} also shows that, as the number of users increases, the ESN-based algorithm can achieve more gain in terms of total successful transmission probability compared to the Q-learning with data correlation. The reason behind this is that the ESN-based transfer RL algorithm uses an approximation method to record historical information while Q-learning uses a Q-table to record the historical information. In consequence, the ESN-based transfer RL algorithm can record more historical information compared to Q-learning and, hence, it can accurately predict the total successful transmission probability. Fig. \ref{figure2} also shows that, as the number of users increases, the gap between Q-learning with data correlation and Q-learning without data correlation increases. The main reason behind this is that, as the number of users increases, the probability that the users request the same content increases and, hence, the data correlation among the users increases.

 \begin{figure}[!t]
  \begin{center}
   \vspace{0cm}
    \includegraphics[width=11cm]{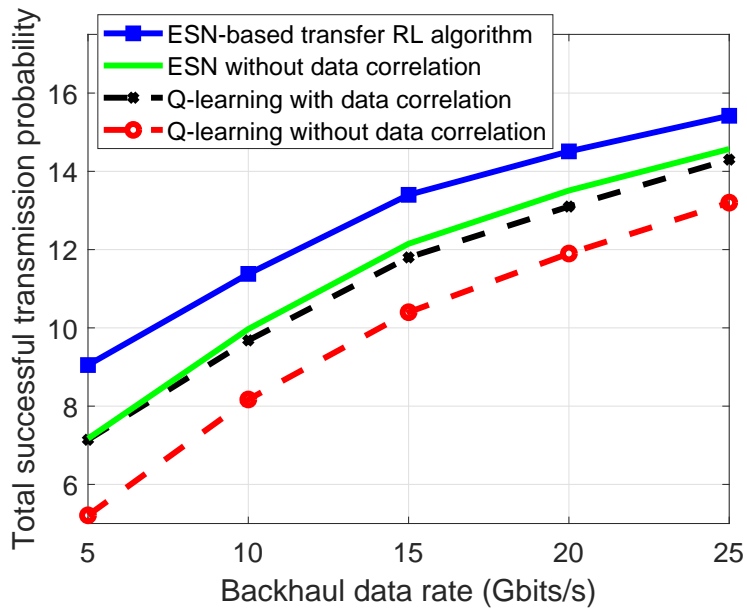}
    \vspace{-0.3cm}
    \caption{\label{figure4} Total successful transmission probability as the backhaul data rate varies.}
  \end{center}\vspace{-1cm}
\end{figure}

Fig. \ref{figure4} shows how the total successful transmission probability changes as the total data rate of the backhaul varies. From Fig. \ref{figure4}, we can see that, as the data rate of the backhaul increases, the total successful transmission probability of all considered algorithms increases. This implies that as the data rate of backhaul increases, the transmission delay over backhaul links decreases. Fig. \ref{figure4} also shows that the gap between the proposed ESN-based transfer RL algorithm and Q-learning with data correlation decreases. This is due to the fact that, as the data rate of the backhaul increases, the data rate of each user for content transmission increases. In consequence, the effect of using data correlation to reduce the data traffic over backhaul links decreases.

 \begin{figure}[!t]
  \begin{center}
   \vspace{0cm}
    \includegraphics[width=11cm]{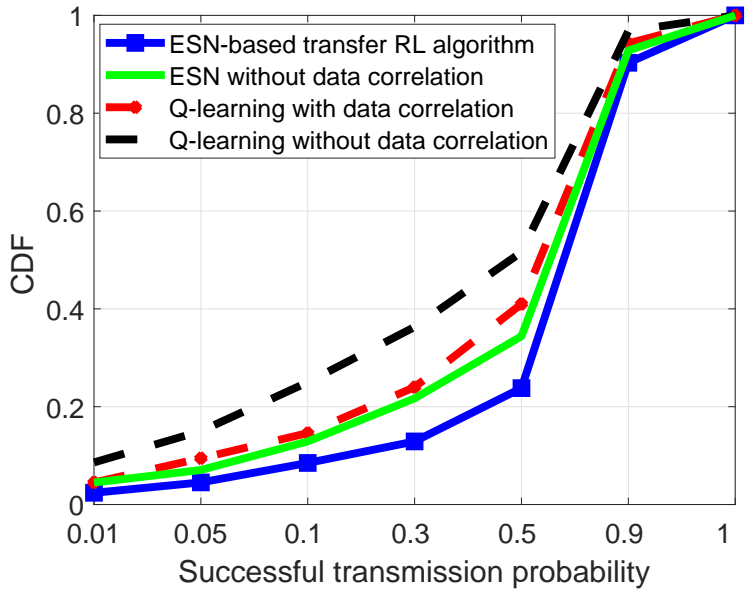}
    \vspace{-0.3cm}
    \caption{\label{figure5} CDFs of the delay resulting from the different algorithms.}
  \end{center}\vspace{-1cm}
\end{figure}

In Fig. \ref{figure5}, we show the cumulative distribution function (CDF) for the total successful transmission probability resulting from all of the considered algorithms. Fig. \ref{figure5} shows that the ESN-based transfer RL algorithm improves the CDF of by to 45\% and 56\% compared, respectively, to Q-learning with data correlation and Q-learning without data correlation at a successful transmission probability of 0.5. These gains stem from the fact that the proposed ESN-based transfer RL algorithm can record more historical information related to the states of the network and users compared to Q-learning. Moreover, the ESN-based transfer RL algorithm can transfer the resource allocation schemes that have been learned in the previous period for the new resource allocation schemes that must be learned in the next period. In consequence, the proposed ESN-based transfer RL algorithm can predict the successful transmission probability more accurately compared to Q-learning and find a better solution for the successful transmission probability maximization.

{\color{black}
 \begin{figure}
  \begin{center}
   \vspace{0cm}
    \includegraphics[width=11cm]{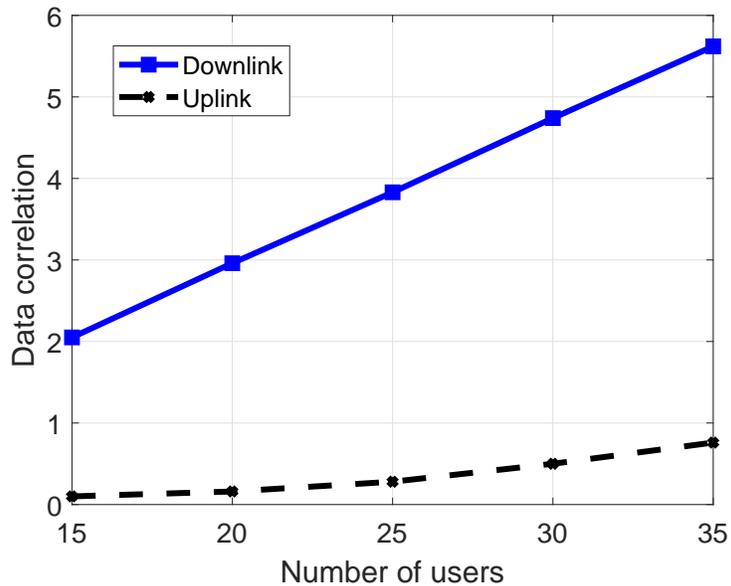}
    \vspace{-0.3cm}
    \color{black}\caption{\label{figure9} Data correlation as the number of users varies.}
  \end{center}\vspace{-0.7cm}
\end{figure}

In Fig. \ref{figure9}, we show how the data correlations of uplink tracking information and downlink VR contents {\color{black} of user $i$} change as the number of VR users varies. {\color{black}In this figure, user $i$ is randomly choosen from the network. The uplink data correlation is  $\sigma _{i}^{\max}=\mathop {\max }\limits_{k\in \mathcal{U}_j, k \ne i} \left( {{\sigma _{ik}}} \right)$, which represents {the maximum covariance for each user $i$.} The downlink data correlation is ${{U_{ja}} - \sum\limits_{n = 2}^{{U_{ja}}} {\sum\limits_{{C_a} \in \mathcal{C}_a^n} {{{\left( { - 1} \right)}^{n - 1}}{C_a}} }    }$, which represents the data correlation among the users that request content $a$. For uplink, we consider the data correlation between only two users while for downlink, we consider the data correlation among multiple VR users.} From this figure, we can see that, as the number of users increases, both the data correlations over uplink and downlink increase. An increase in uplink data correlation is because the distance between two users decreases and the maximum covariance increases as the number of users increases. An increase in downlink data correlation is due to the fact that the number of users that request the same contents increases. Fig. \ref{figure9} also shows that the value of uplink data correlation is below 1 while the downlink data correlation is larger than 1. This is due to the fact that, in uplink, the data correlation is considered between only two users. However, in downlink, the data correlation is considered among multiple users.}

 \begin{figure}
  \begin{center}
   \vspace{0cm}
    \includegraphics[width=11cm]{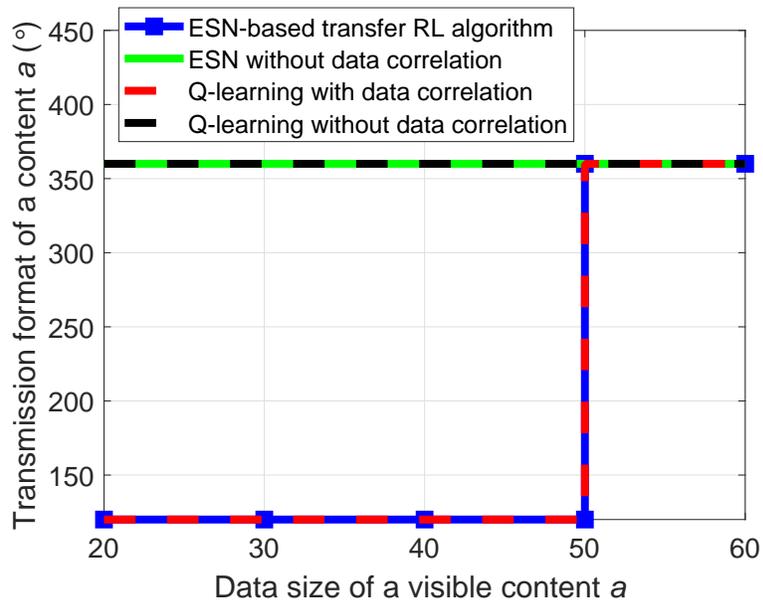}
    \vspace{-0.3cm}
    \caption{\label{figure6} Transmission format of a given content as the data size of the visible contents that the cloud needs to transmit varies.}
  \end{center}\vspace{-0.5cm}
\end{figure}

Fig. \ref{figure6} shows how the transmission format of a given content $a$ changes as the data size of visible contents that the cloud needs to transmit to an SBS varies. Here, the change of the data size of visible contents indicates the change of the data correlation among the users that request content $a$. From Fig. \ref{figure6}, we can see that, as the data size of visible contents does not exceed 50 Mbits, the cloud will transmit visible contents that are extracted from $360^\circ$ content $a$ to the SBS. In contrast, when the data size of visible contents exceeds 50 Mbits, the cloud will transmit $360^\circ$ content $a$ to the SBS. This is because the cloud will always transmit the content that has a smaller data size. Fig. \ref{figure6} also shows that, for the Q-learning algorithm without data correlation, the cloud will always transmit $360^\circ$ contents to the SBS. This is because, for the Q-learning algorithm without data correlation, the cloud will not consider the data correlation among the users.

 \begin{figure}[!t]
  \begin{center}
   \vspace{0cm}
    \includegraphics[width=11cm]{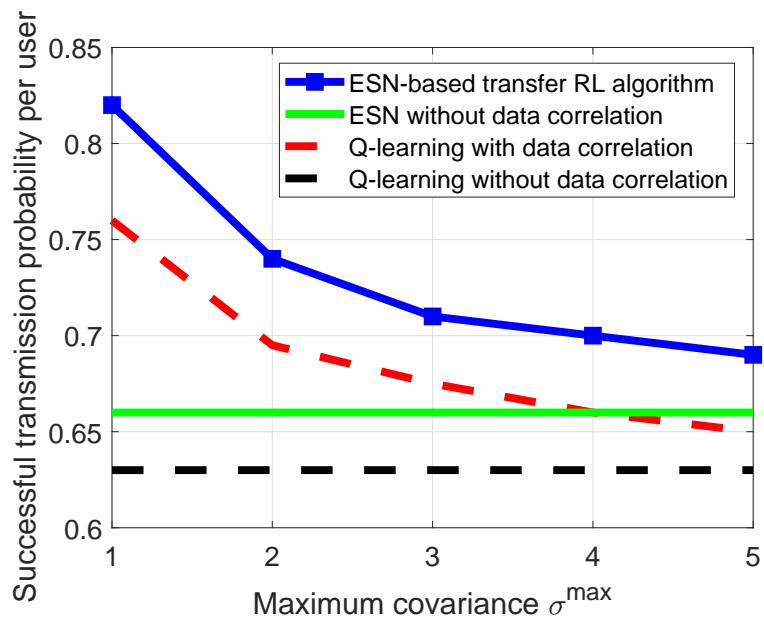}
    \vspace{-0.2cm}
    \caption{\label{figure7} Successful transmission probability per user as $\sigma^{\max}$ varies.}
  \end{center}\vspace{-1cm}
\end{figure}

In Fig. \ref{figure7}, we show how the successful transmission probability per user varies as the uplink maximum covariance $\sigma^{\max}$ changes. Here, an increase in $\sigma^{\max}$ indicates that the maximum uplink data correlation decreases. In contrast, an increase in $\sigma^{\max}$ indicates that the maximum uplink data correlation increases. From Fig. \ref{figure7}, we can see that, as $\sigma^{\max}$ increases (the maximum uplink data correlation decreases), the successful transmission probability per user resulting from all of the considered algorithms decreases. This is due to the fact that, as $\sigma^{\max}$ increases, the data size of the tracking information that the users needs to transmit to the SBS increases. In consequence, the delay of the transmitting tracking information increases. Fig. \ref{figure7} also shows that the gap between the Q-learning with data correlation and Q-learning without data correlation decreases because, as $\sigma^{\max}$ continues to increase, the data correlation of tracking information decreases.       

 \begin{figure}[!t]
  \begin{center}
   \vspace{0cm}
    \includegraphics[width=11cm]{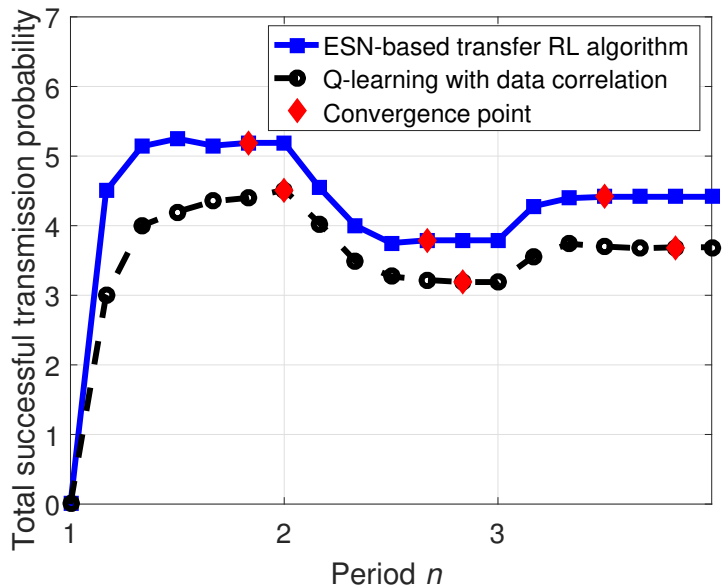}
    \vspace{-0.3cm}
   {\color{black} \caption{\label{convergence} Convergence of learning algorithms.}}
  \end{center}\vspace{-1cm}
\end{figure}

{\color{black} In Fig. \ref{convergence}, we show the number of iterations needed for convergence as period $n$ changes. In this figure, the users' data correlation and content request distribution change with period $n$. {\color{black} The convergence point indicates that each SBS finds its optimal resource allocation vector during each period and the total successful transmission probability of all users associated with this SBS is maximized.} From Fig. \ref{convergence}, we can see that the successful transmission probability per SBS for all considered algorithms increases and, then, converges as time elapses. We can also see that the proposed ESN-based transfer RL algorithm needs $33\%$ less iterations to reach convergence compared to Q-learning with data correlation at period 3. Meanwhile, the proposed algorithm at period 3 uses 33\% less iterations to reach convergence compared to the proposed algorithm at period 1. 
 These gains are due to the fact that the proposed ESN-based transfer RL algorithm can use the already learned successful transmission probability for learning the new successful transmission probability thus increasing the learning speed. }

 \begin{figure}[!t]
  \begin{center}
   \vspace{0cm}
    \includegraphics[width=11cm]{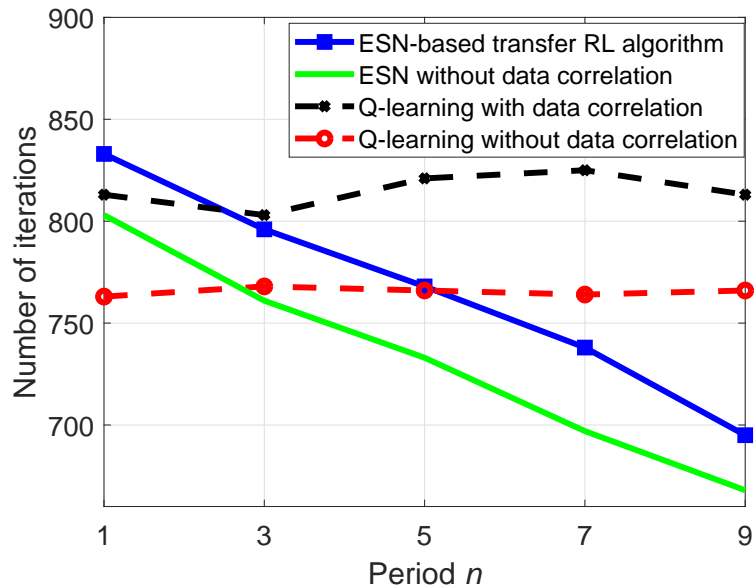}
    \vspace{-0.3cm}
    \caption{\label{figure8} Number of iterations needed for convergence as period varies.}
  \end{center}\vspace{-1cm}
\end{figure}

Fig. \ref{figure8} shows how the number of iterations changes as the period varies. From Fig. \ref{figure8}, we can see that, as the period increases, the number of iterations needed for convergence of the Q-learning with data correlation and Q-learning without data correlation do not change significantly. However, the number of iterations needed for the convergence of the proposed algorithm decreases as the period increases. Fig. \ref{figure8} also shows that the proposed algorithm can achieve up to 9.6\% and 14.3\% gains in terms of the number of iterations needed for convergence compared to the Q-learning with data correlation and Q-learning without data correlation schemes, during period 9. This is due to the fact that, as the period changes, the content request distribution of each user changes. Hence, the data correlation among the users will change and each SBS needs to retrain the Q-learning algorithms so as to find the optimal resource allocation. However, the proposed transfer learning algorithm can directly build the relationship between the actions, states and utility values via transferring the information that was already learned in previous periods. Fig. \ref{figure8} also shows that the number of iterations needed for the convergence of the Q-learning with data correlation is larger than for the Q-learning without data correlation. Meanwhile, the number of iterations needed for the convergence of the Q-learning without data correlation changes slightly as the period changes. This is because the Q-learning without data correlation does not consider the data correlation among the users.       

%
% \begin{figure}[!t]
%  \begin{center}
%   \vspace{0cm}
%    \includegraphics[width=11cm]{figure14.eps}
%    \vspace{-0.3cm}
%   {\color{black} \caption{\label{figure14} An illustrative example of the use of the proposed algorithm for power allocation.}}
%  \end{center}\vspace{-0.6cm}
%\end{figure}
%{\color{black}
%In Fig. \ref{figure14}, we show how the proposed algorithm can jointly optimize the transmit power, resource allocation, and user association so as to maximize the successful transmission probability of each user.  From Fig. 2 of the response, we can see that the proposed algorithm with power allocation can achieve up to 9.8\% gain in terms of successful transmission probability compared to the proposed algorithm with fixed power for a network with 9 SBSs. Fig. 2 also shows that as the number of SBSs increases, the deviation of the total successful transmission probability between the proposed algorithm with power allocation and the proposed algorithm with fixed power increases. This is due to the fact that as the number of SBSs increases, the proposed algorithm with power allocation can adjust the transmit power of each SBS to meet the delay requirement of each user while reducing the interference to the users that are associated with other SBSs.    }

\section{Conclusion}
In this paper, we have studied the problem of resource management in a network of VR users whose data can be correlated. We have formulated this data correlation-aware resource management problem as an optimization problem whose goal is to maximize the VR users' successful transmission probability. To solve this problem, we have developed a neural network reinforcement learning algorithm that uses echo state networks along with transfer learning  to find the most suitable resource block allocations. We have then shown that, by using transfer learning, the proposed algorithm can exploit VR user data correlation to intelligently find the optimal resource allocation strategy as the VR users' content request distribution and data correlation change. Simulation results have evaluated the performance of the proposed approach and shown considerable gains, in terms of total successful transmission probability compared to a classical Q-learning algorithm. {\color{black}Future work can consider additional delay components, such as those related to handover and user association.}

 \section*{Appendix}
 
\subsection{Proof of Theorem \ref{th:2}}\label{Ap:a}   
To simplify the proof, we use $ \boldsymbol{s}, \boldsymbol{v}, D_{t}^\textrm{U}$, and $D_{t}$ to refer to $ \boldsymbol{s}_{ij}, \boldsymbol{v}_{ij}, D_{ijt}^\textrm{U}$, and $D_{ijt}$. For i), the gain that stems from an increase in the allocated uplink resource blocks, $\Delta {\mathbbm{P}_i}$ can be given by:
\begin{equation}\label{eq:DPi1}
\begin{split}
\Delta {\mathbbm{P}_i}&={\mathbbm{P}_i}\left( \boldsymbol{s}, \boldsymbol{v}+\Delta \boldsymbol{v}\right)-{\mathbbm{P}_i}\left( \boldsymbol{s}, \boldsymbol{v} \right)\\
&= \mathop \frac{1}{T}\!\sum\limits_{t = 1}^T  {{\mathbbm{1}_{\left\{ {D_{t}\left(\boldsymbol{s},g_{ja_{it}}\right)+D_{t}^\textrm{U}\left( \boldsymbol{v} +\Delta \boldsymbol{v}, \sigma_{i}^{\max} \right) \le \gamma _D} \right\}}}}-\mathop \frac{1}{T}\!\sum\limits_{t = 1}^T{{\mathbbm{1}_{\left\{ {D_{t}\!\left(\boldsymbol{s},g_{ja_{it}}\right)+ D_{t}^\textrm{U}\left( \boldsymbol{v}, \sigma_{i}^{\max} \right) \le \gamma _D} \right\}}}}.
\end{split}
\end{equation}
From {\color{black}(\ref{eq:DPi1})}, we note that $ {D_{t}^\textrm{U}\left( \boldsymbol{v} +\Delta \boldsymbol{v}, \sigma_{i}^{\max} \right) \le D_{t}^\textrm{U}\left( \boldsymbol{v}, \sigma_i^{\max}\right)}$. This is due to the fact that SBS $j$ allocates more uplink resource blocks $\Delta \boldsymbol{v}$ to user $i$ and, hence, the total delay of user $i$ decreases. 
From {\color{black}(\ref{eq:DPi1})}, we can see that,
\begin{equation}
\Delta {\mathbbm{P}_i}=0,~\textrm{if}~{D_{t}\left(\boldsymbol{s},g_{ja_{it}}\right)+ D_{t}^\textrm{U}\left( \boldsymbol{v},\sigma_i^{\max} \right) \le \gamma _D}. 
\end{equation}
The  is because as ${D_{t}\left(\boldsymbol{s},g_{ja_{it}}\right)+ D_{t}^\textrm{U}\left( \boldsymbol{v},\sigma_i^{\max} \right) \le \gamma _D}$,  ${D_{t}\!\left(\boldsymbol{s},g_{ja_{it}}\right)\!+\!D_{t}^\textrm{U}\!\left( \!\boldsymbol{v} \!+\!\Delta \boldsymbol{v}, \sigma_{i}^{\max} \right) \!\le\! \gamma _D}$. In consequence, as $ {D_{t}\!\left(\boldsymbol{s},g_{ja_{it}}\right)+ D_{t}^\textrm{U}\left( \boldsymbol{v},\sigma_{i}^{\max}\right)\! \le\! \gamma _D}$, $\Delta {\mathbbm{P}_i}=0$.
From (\ref{eq:DPi}), we can also see that,
\begin{equation}
\Delta {\mathbbm{P}_i}=0,~\textrm{if}~{D_{t}\!\left(\boldsymbol{s},g_{ja_{it}}\right)\!+\!D_{t}^\textrm{U}\!\left( \boldsymbol{v} \!+\Delta \boldsymbol{v},\sigma_{i}^{\max} \right) \!> \!\gamma _D}. 
\end{equation}
 This is because as ${D_{t}\!\left(\boldsymbol{s},g_{ja_{it}}\right)\!+\!D_{t}^\textrm{U}\left( \boldsymbol{v} \!+\!\Delta \boldsymbol{v},\sigma_{i}^{\max} \right) \!>\! \gamma _D}$, then $ {D_{t}\!\left(\boldsymbol{s},g_{ja_{it}}\right)\!+\! D_{t}^\textrm{U}\!\left( \boldsymbol{v},\sigma_{i}^{\max}\right)}$. Hence, as ${D_{t}\left(\boldsymbol{s},g_{ja_{it}}\right)+D_{t}^\textrm{U}\left( \boldsymbol{v} +\Delta \boldsymbol{v},\sigma_{i}^{\max} \right) > \gamma _D}$, $\Delta {\mathbbm{P}_i}=0$. Finally, as $D_{t}\left(\boldsymbol{s},g_{ja_{it}}\right)$+$D_{t}^\textrm{U}\left( \boldsymbol{v} +\Delta \boldsymbol{v},\sigma_{i}^{\max} \right) \le \gamma _D$ and $ {D_{t}\!\left(\boldsymbol{s},g_{ja_{it}}\right)+ D_{t}^\textrm{U}\left( \boldsymbol{v},\sigma_{i}^{\max}\right) > \gamma _D}$, $\Delta {\mathbbm{P}_i}$ will be given by:
\begin{equation}
\begin{split}
\Delta {\mathbbm{P}_i}&=\mathop \frac{1}{T}\!\sum\limits_{t = 1}^T  {{\mathbbm{1}_{\left\{ {D_{t}\left(\boldsymbol{s}+\Delta \boldsymbol{s},g_{ja_{it}}\right)+D_{t}^\textrm{U}\left( \boldsymbol{v},\sigma_{i}^{\max} \right) \le \gamma _D} \right\}}}},\\
& \mathop  = \limits^{\left( a \right)}\! \! \mathop \frac{1}{T}\!\sum\limits_{t = 1}^T  \!{{\mathbbm{1}_{\!\!\left\{  \frac{{{K_{i}\left(\sigma_{i}^{\max}\right)}}}{\gamma_D-D_{t}\left(\boldsymbol{s},g_{ja_{it}}\right)}-{c_{ij}\left(\Delta\boldsymbol{v}\right)}  \leqslant  {c_{ij}\left(\boldsymbol{v}_{ij}\right)}  < \frac{{{K_{i}\left(\sigma_{i}^{\max}\right)}}}{\gamma_D-D_{t}\left(\boldsymbol{s},g_{ja_{it}}\right)} \!\right\}}}},
\end{split}
\end{equation}
 where (a) is obtained from the fact that, when ${D_{t}\left(\boldsymbol{s},g_{ja_{it}}\right)+D_{t}^\textrm{U}\left( \boldsymbol{v} +\Delta \boldsymbol{v},\sigma_{i}^{\max} \right) \le \gamma _D}$, ${c_{ij}\left(\boldsymbol{v}\right)} \geqslant \frac{{{K_{i}\left(\sigma_{i}^{\max}\right)}}}{\gamma_D-D_{t}\left(\boldsymbol{s},g_{ja_{it}}\right)}-{c_{ij}\left(\Delta\boldsymbol{v}\right)}$. Similarly, as $ {D_{t}\!\left(\boldsymbol{s},g_{ja_{it}}\right)\!+ \!D_{t}^\textrm{U}\left( \boldsymbol{v},\sigma_{i}^{\max}\right) \!>\! \gamma _D}$ and ${c_{ij}\left(\boldsymbol{v}\right)} < \frac{{{K_{i}\left(\sigma_{i}^{\max}\right)}}}{\gamma_D-D_{t}\left(\boldsymbol{s},g_{ja_{it}}\right)}$. {\color{black}The proof of case ii) is similar to the proof of case i).}

For case iii), we first need to determine the size of visible and $360^\circ$ contents. If $M\left(120^\circ\right)<M\left(360^\circ\right)$, the gain due to the change of content transmission format, $g_{ja}$ can be given by:
\begin{equation}\label{eq:DPi}
\begin{split}
\Delta {\mathbbm{P}_i}&= \mathop \frac{1}{T}\!\sum\limits_{t = 1}^T{{\mathbbm{1}_{\left\{ {D_{t}\!\left(\boldsymbol{s},120^\circ\right)+ D_{t}^\textrm{U}\left( \boldsymbol{v}, \sigma_{i}^{\max} \right) \le \gamma _D} \right\}}}}-\mathop \frac{1}{T}\!\sum\limits_{t = 1}^T  {{\mathbbm{1}_{\left\{ { D_{t}\left(\boldsymbol{s}, 360^\circ\right)+D_{t}^\textrm{U}\left( \boldsymbol{v}, \sigma_{i}^{\max} \right) \le \gamma _D} \right\}}}}.
\end{split}
\end{equation}
Here, since $M\left(120^\circ\right)<M\left(360^\circ\right)$, then $D_{t}\left(\boldsymbol{s}, 360^\circ\right)>D_{t}\!\left(\boldsymbol{s},120^\circ\right)$. In consequence, when ${ D_{t}\left(\boldsymbol{s}, 360^\circ\right)+D_{t}^\textrm{U}\left( \boldsymbol{v} , \sigma_{i}^{\max} \right) \le \gamma _D}$, ${D_{t}\!\left(\boldsymbol{s},120^\circ\right)+ D_{t}^\textrm{U}\left( \boldsymbol{v}, \sigma_{i}^{\max} \right) < \gamma _D}$. Hence,
 \begin{equation}
{{\mathbbm{1}_{\left\{ {D_{t}\left(\boldsymbol{s}, 360^\circ \right)+D_{t}^\textrm{U}\left( \boldsymbol{v} , \sigma_{i}^{\max} \right) \le \gamma _D} \right\}}}}-{{\mathbbm{1}_{\left\{ {D_{t}\!\left(\boldsymbol{s},120^\circ\right)+ D_{t}^\textrm{U}\left( \boldsymbol{v}, \sigma_{i}^{\max} \right) \le \gamma _D} \right\}}}}=0. 
\end{equation} 
Similarly, If ${D_{t}\!\left(\boldsymbol{s},120^\circ\right)+ D_{t}^\textrm{U}\left( \boldsymbol{v}, \sigma_{i}^{\max} \right)  \geqslant  \gamma _D}$, then ${ D_{t}\left(\boldsymbol{s}, 360^\circ\right)+D_{t}^\textrm{U}\left( \boldsymbol{v} , \sigma_{i}^{\max} \right) > \gamma _D}$. Hence, 
 \begin{equation}
{{\mathbbm{1}_{\left\{ {D_{t}\left(\boldsymbol{s}, 360^\circ \right)+D_{t}^\textrm{U}\left( \boldsymbol{v}, \sigma_{i}^{\max} \right) \le \gamma _D} \right\}}}}-{{\mathbbm{1}_{\left\{ {D_{t}\!\left(\boldsymbol{s},120^\circ\right)+ D_{t}^\textrm{U}\left( \boldsymbol{v}, \sigma_{i}^{\max} \right) \le \gamma _D} \right\}}}}=0. 
\end{equation}
Finally, if ${D_{t}\!\left(\boldsymbol{s},120^\circ\right)+ D_{t}^\textrm{U}\left( \boldsymbol{v}, \sigma_{i}^{\max} \right)   \leqslant   \gamma _D}$ and ${ D_{t}\left(\boldsymbol{s}, 360^\circ\right)+D_{t}^\textrm{U}\left( \boldsymbol{v} , \sigma_{i}^{\max} \right) > \gamma _D}$, then $\Delta {\mathbbm{P}_i}$ can be given by:
\begin{equation}\label{eq:DPi3}
\begin{split}
\Delta {\mathbbm{P}_i}&=\mathop \frac{1}{T}\!\sum\limits_{t = 1}^T{{\mathbbm{1}_{\left\{ {D_{t}\!\left(\boldsymbol{s},120^\circ\right)+ D_{t}^\textrm{U}\left( \boldsymbol{v}, \sigma_{i}^{\max} \right) \le \gamma _D} \right\}}}}=\mathop \frac{1}{T}\!\sum\limits_{t = 1}^T{{\mathbbm{1}_{\left\{   {{{M\left( 120^\circ \right)}}}  \leqslant  \left( \gamma_D- \frac{{{ G_{120^\circ} }}}{{{c_{ij}\left(\boldsymbol{s}\right)}}}    -D_{t}^\textrm{U}\left( \boldsymbol{v},\sigma_{i}^{\max} \right) \right){V_{i}^\textrm{B}}\right\}}}},
\end{split}
\end{equation}
(\ref{eq:DPi3}) is hold as ${ D_{t}\left(\boldsymbol{s}, 360^\circ\right)+D_{t}^\textrm{U}\left( \boldsymbol{v} , \sigma_{i}^{\max} \right) > \gamma _D}$. In consequence, $  {{{M\left( 360^\circ \right)}}}$ will be larger than $\left( \gamma_D- \frac{{{ G_{120^\circ} }}}{{{c_{ij}\left(\boldsymbol{s}  \right)}}}    -D_{t}^\textrm{U}\left( \boldsymbol{v},\sigma_{i}^{\max} \right) \right){V_{i}^\textrm{B}}$. We can use the same method to prove the case as $M\left( 120^\circ \right)>M\left( 360^\circ \right)$. When $M\left(120^\circ \right)=M\left(360^\circ \right)$, we can obtain that 
\begin{equation}
{{\mathbbm{1}_{\left\{ {D_{t}\left(\boldsymbol{s}, 360^\circ \right)+D_{t}^\textrm{U}\left( \boldsymbol{v}, \sigma_{i}^{\max} \right) \le \gamma _D} \right\}}}}-{{\mathbbm{1}_{\left\{ {D_{t}\!\left(\boldsymbol{s},120^\circ\right)+ D_{t}^\textrm{U}\left( \boldsymbol{v}, \sigma_{i}^{\max} \right) \le \gamma _D} \right\}}}}=0. 
\end{equation}
Therefore, $\Delta {\mathbbm{P}_i}=0$. This completes the proof.

\subsection{Proof of Theorem \ref{th:1}}\label{Ap:a}  
To prove Theorem \ref{th:1}, we must first calculate the data size  $L_a\left(\mathcal{C}_{a}\right)$ of visible contents extracted from $360^\circ$ content $a$ using an enumeration method. Let $U_{ja}=2$. If users $i$ and $j$ that are associated with SBS $j$ request visible contents $a_i$ and $a_j$, respectively, the set of the data correlation between these two users is ${C}_{aik}$. This means that users $i$ and $k$ have ${C}_{aik}$ portion of visible content that is similar. In consequence, when the cloud transmits visible contents $a_i$ and $a_j$ to SBS $j$, the size of the data that the cloud needs to transmit is given by:
\begin{equation}
L_a\left(\mathcal{C}_{a}\right)=G_{120^\circ}\left( U_{ja} \!-\!{C}_{aik}\right)\!=\!G_{120^\circ}\left( U_{ja} \!-\!{\sum\limits_{{C_a} \in \mathcal{C}_a^2} {{C_a}} } \right),
\end{equation}
 where $\mathcal{C}_{a}=\mathcal{C}_a^2=\left\{  {C}_{aik} \right\}$. Similarly, if users $i$, $j$, and $k$ that are associated with SBS $j$ request visible contents $a_i$, $a_j$, and $a_k$, respectively, ($U_{ja}=3$), the size of the data that the cloud needs to transmit can be given by: 
 \begin{equation}
 L_a\left(\mathcal{C}_{a}\right)=G_{120^\circ}\left( U_{ja} \!-\!{C}_{aik}\!-\!{C}_{aij}-{C}_{ajk}+{C}_{aijk}\right)=G_{120^\circ}\!\left(\! U_{ja} \!-\!\!\!\!{\sum\limits_{{C_a} \in \mathcal{C}_a^2} \!\!{{C_a}} }\!+\!{\sum\limits_{{C_a} \in \mathcal{C}_a^3} \!\!{{C_a}} }\!\!\right), 
 \end{equation}
where $\mathcal{C}_{a}=\mathcal{C}_a^2  \cap \mathcal{C}_a^3$ with $\mathcal{C}_a^2=\left\{ {C}_{aik},{C}_{aij},{C}_{ajk}\right\}$ and $\mathcal{C}_a^3=\left\{ {C}_{aijk}\right\}.$
 We can also see that if $U_{ja}=4$, the cloud needs to transmit is given by:
 \begin{equation}
 L_a\left(\mathcal{C}_{a}\right)=G_{120^\circ}\!\left( U_{ja} -{\sum\limits_{{C_a} \in \mathcal{C}_a^2} \!\!{{C_a}} }\!+\!{\sum\limits_{{C_a} \in \mathcal{C}_a^3} \!\!{{C_a}} }-{\sum\limits_{{C_a} \in \mathcal{C}_a^4} {{C_a}} }\right), 
 \end{equation}
 where $\mathcal{C}_{a}=\mathcal{C}_a^2  \cap \mathcal{C}_a^3 \cap \mathcal{C}_a^4$.
 Hence, we can obtain that $L_a\left(\mathcal{C}_{a}\right)= G_{120^\circ} \left( {{U_{ja}} - \sum\limits_{n = 2}^{{U_{ja}}} {\sum\limits_{{C_a} \in \mathcal{C}_a^n} {{{\left( { - 1} \right)}^{n - 1}}{C_a}} } } \right)$.
  The cloud will select the content transmission format that can minimize the size of the data transmitted over cloud-SBS links. In consequence, if $G_{360^\circ}>L_a\left(\mathcal{C}_{a}\right)$, the cloud will transmit visible contents that are extracted from $360^\circ$ content $a$, $G_{360^\circ} \leqslant L_a\left(\mathcal{C}_{a}\right)$, otherwise. This completes the proof. 

\bibliographystyle{IEEEbib}
\def\baselinestretch{0.9}
\bibliography{references}

\begin{thebibliography}{10}

\bibitem{ESTL2017Chen}
M.~Chen, W.~Saad, C.~Yin, and M.~Debbah,
\newblock ``Echo state transfer learning for data correlation aware resource
  allocation in wireless virtual reality,''
\newblock in {\em Proc. of Asilomar Conference on Signals, Systems and
  Computers}, Pacific Grove, CA, USA, Oct. 2017.

\bibitem{bacstuug2016towards}
E.~Ba{\c{s}}tu{\u{g}}, M.~Bennis, M.~M{\'e}dard, and M.~Debbah,
\newblock ``Towards interconnected virtual reality: {O}pportunities, challenges
  and enablers,''
\newblock {\em IEEE Communications Magazine}, vol. 55, no. 6, pp. 110--117,
  Jan. 2017.

\bibitem{rosedale2017virtual}
P.~Rosedale,
\newblock ``Virtual reality: {T}he next disruptor: {A} new kind of worldwide
  communication,''
\newblock {\em IEEE Consumer Electronics Magazine}, vol. 6, no. 1, pp. 48--50,
  Jan. 2017.

\bibitem{ahn2017delay}
J.~Ahn, Y.~Yong Kim, and R.~Y. Kim,
\newblock ``Delay oriented {VR} mode {WLAN} for efficient wireless multi-user
  virtual reality device,''
\newblock in {\em Proc. of IEEE International Conference on Consumer
  Electronics}, Las Vegas, NV, USA, March 2017.

\bibitem{singh2017high}
M.~Singh and B.~Jung,
\newblock ``High-definition wireless personal area tracking using {AC} magnetic
  field for virtual reality,''
\newblock in {\em Proc. of IEEE Virtual Reality}, Los Angeles, California, USA,
  March 2017.

\bibitem{chakareski2017vr}
J.~Chakareski,
\newblock ``{VR/AR} immersive communication: {C}aching, edge computing, and
  transmission trade-offs,''
\newblock in {\em Proc. of the Workshop on Virtual Reality and Augmented
  Reality Network}, Los Angeles, CA, USA, Aug. 2017.

\bibitem{VROWNchen}
M.~Chen, W.~Saad, and C.~Yin,
\newblock ``Virtual reality over wireless networks: Quality-of-service model
  and learning-based resource management,''
\newblock {\em IEEE Transactions on Communications}, to appear, 2018.

\bibitem{kasgari2018human}
A.~Taleb~Zadeh Kasgari, W.~Saad, and M.~Debbah,
\newblock ``Human-in-the-loop wireless communications: {M}achine learning and
  brain-aware resource management,''
\newblock {\em arXiv preprint arXiv:1804.00209}, March 2018.

\bibitem{park2018urllc}
J.~Park and M.~Bennis,
\newblock ``{URLLC-eMBB} slicing to support {VR} multimodal perceptions over
  wireless cellular systems,''
\newblock {\em available online: arxiv.org/abs/1805.00142}, May 2018.

\bibitem{sun2018communication}
Y.~Sun, Z.~Chen, M.~Tao, and H.~Liu,
\newblock ``Communication, computing and caching for mobile {VR} delivery:
  {M}odeling and trade-off,''
\newblock {\em arXiv preprint arXiv:1804.10335}, April 2018.

\bibitem{8319985}
X.~Yang, Z.~Chen, K.~Li, Y.~Sun, N.~Liu, W.~Xie, and Y.~Zhao,
\newblock ``Communication-constrained mobile edge computing systems for
  wireless virtual reality: {S}cheduling and tradeoff,''
\newblock {\em IEEE Access}, vol. 6, pp. 16665--16677, March 2018.

\bibitem{8377419}
M.~S. Elbamby, C.~Perfecto, M.~Bennis, and K.~Doppler,
\newblock ``Edge computing meets millimeter-wave enabled {VR}: {P}aving the way
  to cutting the cord,''
\newblock in {\em Proc. of IEEE Wireless Communications and Networking
  Conference}, Barcelona, Spain, April 2018.

\bibitem{5288526}
S.~J. Pan and Q.~Yang,
\newblock ``A survey on transfer learning,''
\newblock {\em IEEE Transactions on Knowledge and Data Engineering}, vol. 22,
  no. 10, pp. 1345--1359, Oct 2010.

\bibitem{chen2017machine}
M.~Chen, U.~Challita, W.~Saad, C.~Yin, and M.~Debbah,
\newblock ``Machine learning for wireless networks with artificial
  intelligence: A tutorial on neural networks,''
\newblock {\em available online: arxiv.org/abs/1710.02913}, Oct. 2017.

\bibitem{htc}
HTC,
\newblock ``{HTC} vive,'' \url{ https://www.vive.com/us/}.

\bibitem{8382257}
Z.~Yang, C.~Pan, Y.~Pan, Y.~Wu, W.~Xu, M.~Shikh-Bahaei, and M.~Chen,
\newblock ``Cache placement in two-tier {HetNets} with limited storage
  capacity: {C}ache or buffer?,''
\newblock {\em IEEE Transactions on Communications}, vol. 66, no. 11, pp.
  5415--5429, Nov. 2018.

\bibitem{Wu:2017:DEU:3083187.3083210}
C.~Wu, Z.~Tan, Z.~Wang, and S.~Yang,
\newblock ``A dataset for exploring user behaviors in {VR} spherical video
  streaming,''
\newblock in {\em Proc. of the ACM on Multimedia Systems Conference}, New York,
  NY, USA, 2017.

\bibitem{cressie2015statistics}
N.~Cressie and C.~K. Wikle,
\newblock {\em Statistics for Spatio-Temporal Data},
\newblock John Wiley \& Sons, 2015.

\bibitem{Bishop2006Pattern}
C.~M. Bishop and N.~M. Nasrabadi,
\newblock {\em Pattern Recognition and Machine Learning},
\newblock Springer,, 2006.

\bibitem{vuran2004spatio}
M.~C. Vuran, {O}.~B. Akan, and I.~F. Akyildiz,
\newblock ``Spatio-temporal correlation: theory and applications for wireless
  sensor networks,''
\newblock {\em Computer Networks}, vol. 45, no. 3, pp. 245--259, June 2004.

\bibitem{8254516}
K.~Hamidouche, W.~Saad, and M.~Debbah,
\newblock ``Popular matching games for correlation-aware resource allocation in
  the internet of things,''
\newblock in {\em Proc. of IEEE Global Communications Conference}, Singapore,
  Dec 2017.

\bibitem{7438747}
Z.~Zhao, M.~Peng, Z.~Ding, W.~Wang, and H.~V. Poor,
\newblock ``Cluster content caching: An energy-efficient approach to improve
  quality of service in cloud radio access networks,''
\newblock {\em IEEE Journal on Selected Areas in Communications}, vol. 34, no.
  5, pp. 1207--1221, May 2016.

\bibitem{26}
I.~Szita, V.~Gyenes, and A.~L\H{o}rincz,
\newblock ``Reinforcement learning with echo state networks,''
\newblock {\em Lecture Notes in Computer Science}, vol. 4131, pp. 830--839,
  2006.

\bibitem{chen2016caching}
M.~Chen, M.~Mozaffari, W.~Saad, C.~Yin, M.~Debbah, and C.~S. Hong,
\newblock ``Caching in the sky: {P}roactive deployment of cache-enabled
  unmanned aerial vehicles for optimized quality-of-experience,''
\newblock {\em IEEE Journal on Selected Areas on Communications (JSAC)}, vol.
  35, no. 5, pp. 1046--1061, May 2017.

\bibitem{APractical}
M.~Luko\u{s}evicius,
\newblock {\em A Practical Guide to Applying Echo State Networks},
\newblock Springer Berlin Heidelberg, 2012.

\bibitem{chu2018reinforcement}
M.~Chu, H.~Li, X.~Liao, and S.~Cui,
\newblock ``Reinforcement learning based multi-access control and battery
  prediction with energy harvesting in {IoT} systems,''
\newblock {\em arXiv preprint arXiv:1805.05929}, May 2018.

\bibitem{30}
M.~Bennis and D.~Niyato,
\newblock ``A {Q}-learning based approach to interference avoidance in
  self-organized femtocell networks,''
\newblock in {\em Proc. of IEEE Global Commun. Conference (GLOBECOM) Workshop
  on Femtocell Networks}, Miami, FL, USA, Dec. 2010.

\bibitem{8359094}
U.~Challita, L.~Dong, and W.~Saad,
\newblock ``Proactive resource management for {LTE} in unlicensed spectrum: {A}
  deep learning perspective,''
\newblock {\em IEEE Transactions on Wireless Communications}, vol. 17, no. 7,
  pp. 4674--4689, July 2018.

\bibitem{8114362}
Z.~Yang, C.~Pan, W.~Xu, Y.~Pan, M.~Chen, and M.~Elkashlan,
\newblock ``Power control for multi-cell networks with non-orthogonal multiple
  access,''
\newblock {\em IEEE Transactions on Wireless Communications}, vol. 17, no. 2,
  pp. 927--942, Feb 2018.

\bibitem{3gpp.36.814}
3GPP,
\newblock ``{Evolved Universal Terrestrial Radio Access (E-UTRA); Further
  advancements for E-UTRA physical layer aspects},''
\newblock Technical Report (TR) 36.814, {3rd Generation Partnership Project
  (3GPP)}, 03 2017,
\newblock Version 9.2.0.

\bibitem{Mozaffari2016Unmanned}
M.~Mozaffari, W.~Saad, M.~Bennis, and M.~Debbah,
\newblock ``Unmanned aerial vehicle with underlaid device-to-device
  communications: Performance and tradeoffs,''
\newblock {\em IEEE Transactions on Wireless Communications}, vol. 15, no. 6,
  pp. 3949--3963, June 2016.

\end{thebibliography}
\end{document}